\documentclass[a4paper,cleveref, autoref, thm-restate]{lipics-v2021}
\hideLIPIcs  %uncomment to remove references to LIPIcs series (logo, DOI, ...), e.g. when preparing a pre-final version to be uploaded to arXiv or another public repository
\usepackage{environ}
\NewEnviron{killcontents}{}
\usepackage{todonotes}
\presetkeys%
{todonotes}%
{inline,backgroundcolor=yellow}{}
\usepackage{graphicx}
\usepackage{amsmath, amssymb, amsthm}
\usepackage{pifont}
\usepackage{bbold}
\usepackage{float}
\usepackage{mathtools}
\usepackage{ebproof}
\usepackage{braket}
\usepackage{comment}
\usepackage{xcolor}
\usepackage{float}
\usepackage{mfirstuc}
\usepackage{appendix}
\usepackage{thmtools, thm-restate}
\usepackage{xparse}
\usepackage[ruled, lined, noend]{algorithm2e}
\usepackage{array}
\usepackage{ebproof}
\usepackage{stmaryrd}
\hypersetup{colorlinks,linkcolor={blue},citecolor={blue}}
\newcommand{\pt}[1]{
  \begin{prooftree}[center = false] #1
\end{prooftree}}
\newcommand{\ptsc}[1]{
  \begin{prooftree}[center=false, right label template = \small\inserttext,
    template = \small$\inserttext$] #1
\end{prooftree}}
\ebproofnewstyle{small}{
  rule margin = 1.5ex,
  separation = 0.5em,
  right label template = \footnotesize\inserttext,
  left label template = \footnotesize\inserttext,
  template = \footnotesize$\inserttext$ 
}

\newcommand{\cvar}{C}
\newcommand{\tvar}{T}
\newcommand{\bqvar}{B_Q}
\newcommand{\bcvar}{B_C}
\newcommand{\bvar}{B}
\newcommand{\qvar}{Q}
\newcommand\bkvar{\bvar^{\ket{}}}
\newcommand{\bkvarb}{{\bkvar}}
\newcommand{\qbit}{\mathrm{Qbit}}
\newcommand{\nonlinfunc}{\Rightarrow}
\newcommand{\linfunc}{\multimap}
\newcommand{\anyfunc}{\rightarrowtail}

\newcommand{\bset}{\mathbb{B}}
\newcommand{\bcset}{\mathbb{B}_C}
\newcommand{\bqset}{\mathbb{B}_Q}
\newcommand{\unit}{\mathbb 1}
\newcommand{\nat}{\mathtt{nat}}
\newcommand{\typelist}[1]{\mathtt{list}(#1)}
\newcommand{\cons}[1]{\mathrm{Cons}(#1)}
\newcommand{\sign}[2]{\ifx&#1&
  \else
  #1 \to
  \fi
#2}
\newcommand{\signdefault}{\sign {\bkvar_1, \dots, \bkvar_n} \bvar}
\newcommand{\vunit}{()}
\newcommand{\nil}{[~]}

\newcommand{\textqcase}{\ensuremath{\mathtt{qcase}}}
\newcommand{\textmatch}{\ensuremath{\mathtt{match}}}
\newcommand{\textletrec}{\ensuremath{\mathtt{letrec}}}
\newcommand{\textshape}{\ensuremath{\mathtt{shape}}}
\newcommand{\zket}{\ket 0}
\newcommand{\oket}{\ket 1}
\newcommand{\qcasebase}[2]{
  \textqcase\,#1\left\{\! \begin{aligned} #2 \end{aligned}\right\}
}
\newcommand{\qcase}[3]{
  \qcasebase {#1} {\zket \to #2\,, \oket \to #3}
}
\newcommand{\qcasesplit}[3]{
  \qcasebase {#1} { 
    \begin{aligned} 
      \zket &\to #2 \\ \oket &\to #3 
    \end{aligned}
  }
}
\newcommand{\qcasedefault}[1]{\ensuremath{\qcase{#1}{t_0}{t_1}}}
\newcommand{\match}[2]{
  \textmatch\,#1\left\{\! \begin{aligned} #2 \end{aligned}\right\}
}
\newcommand{\matchdefault}[1]{\textmatch\,#1 \left\{
    c_1(\overrightarrow{x_1}) \to t_1\,, \dots,
    c_n(\overrightarrow{x_n}) \to t_n
\right\}}
\newcommand{\matchdefaultcollapsedany}[2]{
  \textmatch_{1 \leq i \leq n}\,#1\,\{
    c_i(\overrightarrow{x_i}) \to #2_i
  \}
}
\newcommand{\matchdefaultcollapsed}[1]{\matchdefaultcollapsedany{#1}{t}}
\newcommand{\lbd}[2]{\lambda #1 .#2}
\newcommand{\letrec}[3]{\textletrec\,#1\,#2=#3}
\newcommand{\shape}[1]{\mathtt{shape}(#1)}

\newcommand{\can}{\ensuremath{\mathtt{CAN}}}
\newcommand{\equivcontext}{\ensuremath{C_\equiv}}
\newcommand{\fv}[1]{\mathrm{FV}(#1)}

\newcommand{\valueset}{\mathrm{Value}}
\newcommand{\subterm}{\vartriangleleft}
\newcommand{\subtermeq}{\trianglelefteq}

\newcommand{\vartyping}[2]{#1 : #2}
\newcommand{\typing}[4]{
  \ifx&#1&\varnothing \else #1\fi; 
  \ifx&#2&\varnothing \else #2\fi 
  \vdash {#3} : #4
}
\newcommand{\closedtyping}[2]{\typing{}{}{#1}{#2}}
\newcommand{\shapemarker}[1]{[#1]}

\newcommand{\typedom}[1]{\mathrm{dom}(#1)}
\newcommand{\pitwo}{\Pi_2^0}
\newcommand{\portho}{\mathtt{ORTHO}}
\newcommand{\uhalt}{\mathtt{UHalt}}
\newcommand\kron[2]{\delta_{#1, #2}}

\newcommand\innerprod[2]{\langle #1, #2 \rangle}

\newcommand{\subarrow}{\to}
\newcommand{\subsupp}[1]{\mathrm{Supp}(#1)}

\newcommand{\reduces}{\rightsquigarrow}
\newcommand{\sreduces}{\reduces^*}
\newcommand{\evalcontext}{\ensuremath{E}}
\newcommand{\zoreduces}{\ensuremath{\reduces^?}}%Zero or one reduction
\newcommand{\mostreduces}[1]{\ensuremath{\reduces^{\leq #1}}}

\newcommand{\fbqpscript}{\mathtt{QList}}
\newcommand{\circuittype}{\ensuremath{\mathtt{CircuitTerms}}}
\newcommand{\hyrqlbound}[1]{\ensuremath{\mathtt{Hyrql}({#1})}}
\newcommand{\hyrqlpoly}{\hyrqlbound{\mathrm{Poly}}_\fbqpscript}
\newcommand{\binfunc}{\set{0,1}^* \to \set{0,1}^*}
\newcommand{\boundedrec}{\ensuremath{\mathtt{Faithful}}}
\newcommand{\shapeset}[1]{\mathtt{S}(#1)}
\newcommand{\famcirc}{\mathtt{C}(t)}
\newcommand{\bqwalk}{\mathtt{bqwalk}}
\newcommand{\timeset}[1]{\mathtt{Time}(#1)}

\newcommand{\ie}{\text{i.e., }}
\newcommand{\eg}{\text{e.g., }}
\newcommand{\gmid}{\,\mid\,}
\newcommand{\size}[1]{|#1|}
\newcommand{\QS}{\ensuremath{\mathtt{QS}}}
\newcommand{\hyrql}{{{\tt{Hyrql}}}}
\newcommand{\spm}{{{\tt{Spm}}}}
\newcommand{\mcal}{\mathcal}
\newcommand{\seq}[1]{\overrightarrow{#1}}
\newcommand{\fbqp}{\ensuremath{\mathtt{FBQP}}}
\newcommand{\bigo}[1]{\mathcal O({#1})}
\newcommand{\mbN}{\mathbb N}
\newcommand{\mtf}{\mathtt f}
\newcommand{\mcC}{\mathcal C}
\newcommand{\computes}[1]{\llbracket #1 \rrbracket}
\newcommand{\qt}[2]{\theta_{#1}(#2)}
\newcommand{\kroneq}[2]{\nu_{#1, #2}}

\title{Resource-Aware Quantum Programming with General
Recursion and Quantum Control}
\titlerunning{A Resource-Aware Hybrid Language}

\author{Kostia Chardonnet}{Université de Lorraine, CNRS, Inria, LORIA, F-54000 Nancy, France}{kostia.chardonnet@inria.fr}{0009-0000-0671-6390}{}
\author{Emmanuel Hainry}{Université de Lorraine, CNRS, Inria, LORIA, F-54000 Nancy, France}{emmanuel.hainry@loria.fr}{0000-0002-9750-0460}{}
\author{Romain Péchoux}{Université de Lorraine, CNRS, Inria, LORIA, F-54000 Nancy, France}{romain.pechoux@loria.fr}{0000-0003-0601-5425}{}
\author{Thomas Vinet}{Université de Lorraine, CNRS, Inria, LORIA, F-54000 Nancy, France}{thomas.vinet@inria.fr}{0009-0007-8547-6145}{}
\authorrunning{K. Chardonnet et al.}

\Copyright{Jane Open Access and Joan R. Public} 



\nolinenumbers %uncomment to disable line numbering

\EventEditors{John Q. Open and Joan R. Access}
\EventNoEds{2}
\EventLongTitle{42nd Conference on Very Important Topics (CVIT 2016)}
\EventShortTitle{CVIT 2016}
\EventAcronym{CVIT}
\EventYear{2016}
\EventDate{December 24--27, 2016}
\EventLocation{Little Whinging, United Kingdom}
\EventLogo{}
\SeriesVolume{42}
\ArticleNo{23}
\begin{document}

\maketitle

\begin{abstract}
This paper introduces the hybrid quantum language with general recursion
$\hyrql$, driven towards resource-analysis. By design, $\hyrql$ does not
require the specification of an initial set of quantum gates. Hence, it  is
well amenable towards a generic cost analysis, unlike languages that use
different sets of quantum gates, which yield quantum circuits of distinct
complexity.

Regarding resource-analysis, we show how to relate the runtime of an expressive
fragment of $\hyrql$ programs with the size of the corresponding quantum
circuits. We also manage to capture the class of functions computable in
quantum polynomial time, which, by Yao's Theorem, corresponds to families of
circuits of polynomial size. Consequently, this result paves the way for the
use of termination and runtime-analysis techniques designed for classical
programs to guarantee bounds on the size of quantum circuits.
\end{abstract}

\section{Introduction}
\label{sec:intro}
\textbf{Motivations. }
Most well-known quantum algorithms, such as Shor's algorithm~\cite{Sho97}, were
historically designed on the Quantum Random Access Machine (QRAM)
model~\cite{Kni22}. In this model, a program
interacts with a quantum memory through basic operations complying with the
laws of quantum mechanics. These operations include a fixed set of quantum
gates, chosen to be universal, as well as a probabilistic measurement of
qubits. Consequently, the control flow is purely classical, \ie depends on the
(classical) outcome of a measure. Based on  this paradigm, several high-level
quantum programming languages have been introduced, each with different
purposes and applications, from assembly code~\cite{QASM17,QASM22}, to
imperative languages~\cite{FY21}, circuit description languages~\cite{GLRSV13},
and $\lambda$-calculi~\cite{SV09}.

A relevant issue was to extend these models to programs with \emph{quantum
control}, also known as coherent control, enabling a ``program as data''
treatment for the quantum paradigm. Quantum control consists in the ability to
write superposition of programs in addition to superposition of data and
increases the expressive power of quantum programming languages. It allows the
programmer to write algorithms such as the \emph{quantum switch}, which uses
fewer resources (quantum gates) than algorithms with classical
control~\cite{CAPV13} and is physically implementable~\cite{AWHM20,PMA+15}.
Hence, quantum control provides a computational advantage over classical
control~\cite{ACB14,TCM+21,KOY24}. In the last decades, several  quantum
programming languages implementing this concept have been designed,
non-exhaustively~\cite{AG05,SVV18,DCM22,Lem24, YVC24}. To get the best of both worlds,
a natural next step was the development of \emph{hybrid languages}, \ie
languages that allow classical and quantum flow/data to be combined.
The hybrid paradigm is conveying considerable practical interest: for example,
it is used by quantum variational algorithms, a class of quantum algorithms
leveraging both classical and quantum computing resources to find approximate
solutions to optimization problems; and the properties of hybrid languages have
also been deeply studied, non-exhaustively~\cite{Qunity23,Yin24,DLPZ25,BPP25}.

Since hybrid languages offer interesting prospects in terms of expressiveness
and optimal resource consumption, a relevant and open issue concerns the
development of resource-aware (static) analyses of their programs. These static
analyses can be applied to predict the low-level resources required to execute
quantum algorithms, such as the depth or size of a quantum circuit.

\medskip
\noindent
\textbf{Contributions. }
This work solves the above issue for the first time by introducing a HYbrid
Recursive Quantum Language \hyrql{} on which a resource analysis can be
performed. \hyrql{} is a hybrid extension of the functional quantum language
$\spm$ (\emph{Symmetrical pattern-matching}) of~\cite{SVV18,Lem24}. In $\spm$,
the programmer can write down directly unitary applications, combining quantum
superposition of terms with pattern-matching. Unitarity is enforced by a
(decidable) linear typing discipline using an orthogonality predicate. \spm{}
is, however, a purely quantum language, which makes the manipulation (and,
consequently, resource analysis) of classical information awkward: they can be
neither discarded nor duplicated by linearity. The hybrid nature of \hyrql{}
relies on a typing discipline which delineates a clear separation between
quantum/linear and classical/non-linear data in the typing contexts. Similarly
to $\spm$, $\hyrql$ does not include measurement: the flow from quantum data to
classical data is handled through the use of a \textshape\ construct,  as
introduced in \cite{GLRSV13}, which returns the classic structural information
on data containing quantum information, without getting any information on the
value of the quantum states. For example, the  \textshape\ of a qubit list
provides classical (duplicable) information, such as the number of qubits. The
decision to provide $\hyrql$ with a $\spm$-architecture was motivated by two
important factors. First, in contrast with most of its competitors, $\spm$ does
not require the inclusion of an initial set of quantum gates, making its
resource analysis generic. Second, thanks to its pattern-matching design, $\spm
$ allows for the reuse of a wide variety of tools designed for resource
analysis (e.g., termination or complexity) of term rewriting systems.

Our paper contains the following main contributions:

\begin{itemize}
  \item The introduction of the hybrid language with general recursion $\hyrql$,
    as well as its operational semantics, type system, and
    standard properties: confluence (Theorem~\ref{thm:confluence}), progress
    (Theorem~\ref{lem:progress}), and subject reduction
    (Theorem~\ref{lem:subred}).
  \item A proof that the orthogonality predicate in $\hyrql$ is
    $\Pi^0_2$-complete, hence undecidable
    (Theorem~\ref{thm:orthogonality-undecidable}). Consequently, typing is not
    decidable. This is not surprising as the language encompasses general
    recursion. The decidability of typing can still be recovered, on expressive
    sub-fragments of the language, e.g., when $\hyrql$ is restricted to finite
    types and terminating programs (Proposition~\ref{prop:orthogonality-lower}).
  \item A compilation from terms (of the good type) to quantum circuits, whose
    size is bounded by the runtime-complexity of the initial term
    (Theorem~\ref{thm:circuit-bound}). In particular, terms that terminate in
    polynomial time characterize exactly the complexity class \fbqp{}
    (Theorem~\ref{thm:fbqp-sound} and~\ref{thm:fbqp-complete}). This result is
    not trivial, as lambda calculi do not comply with van Emde Boas' invariant
    thesis \cite{VEB90}, i.e., a polynomial time reduction on a lambda term
    may correspond to an exponential time reduction on a Turing machine.

  \item An overview on how complexity results can be obtained and
    automatized by translating our language into term-rewrite systems, to use
    existing techniques for runtime analysis (e.g.,
    interpretations~\cite{MP09}, recursive path orderings~\cite{Der82,Der87},
    dependency pairs~\cite{AG00}, size-change principle \cite{LJBA01}).
\end{itemize}

We illustrate the expressive power of the language through several examples:
implementa\-tion of basic quantum gates (Hadamard gate,
Example~\ref{ex:had-syntax}), higher-order (quantum switch,
Example~\ref{ex:qs-syntax}), general recursion (bounded quantum walk,
Example~\ref{ex:bqwalk}), and hybrid control flow (protocol
\emph{BB84}~\cite{BB84}, Example~\ref{ex:hybrid-function}).

\newcommand{\dlpz}{{\texttt{DLPZ}}}
\newcommand{\foq}{{\texttt{FOQ}}}
\newcommand{\qugcl}{{\texttt{QuGCL}}}
\newcommand{\qunity}{{\texttt{Qunity}}}
\newcommand{\rqc}{{\texttt{RQC}}}

\begingroup
\newcommand{\cmark}{\ding{51}}
\newcommand{\xmark}{\ding{55}}
\renewcommand{\arraystretch}{1.2}
\setlength\tabcolsep{3pt}
\begin{table}[!h]
  \begin{center}
      \begin{tabular}{c | c | c | c | c | c | c }
        & \hyrql{} &\foq{} \cite{HPS23} & $\rqc^{++}$ \cite{YZ24} &
        $\mathtt{Silq}$ \cite{Silq20} & \qunity{} \cite{Qunity23} & \dlpz{}
        \cite{DLPZ25} \\ \hline 
        Resource-aware & \cmark & \cmark & \xmark & \xmark & \xmark & \xmark \\
        \hline
        General recursion & \cmark & \xmark & \cmark & \xmark & \xmark & \xmark
        \\
        Higher-order & \cmark & \xmark & \xmark & \xmark & \xmark & \xmark
      \end{tabular}
  \end{center}
  \caption{Comparison table between hybrid quantum languages }
  \label{tab:related}
  \vspace*{-0.5cm}
\end{table}
\endgroup

\medskip
\noindent
\textbf{Related works. }
In Table~\ref{tab:related}, we provide a non-exhaustive comparison between the
main families of hybrid quantum languages which treat classical data as
first-class citizen.
Towards that end, we consider the three following criteria.
\begin{itemize} 
  \item \textit{Resource-aware} highlights if the language is designed towards
    resource analysis, that is, allows to characterize well-known complex
    classes of quantum computational models. This is the case of
    $\foq$~\cite{HPS23,HPS25}, which characterizes, under restrictions, quantum
    polynomial time and quantum polylogarithmic time~\cite{FHPS25}.
    However, $\hyrql$ has a greater expressive power than $\foq$  as it is not
    restricted to unitary operators. Moreover, all standard inductive datatypes
    can be expressed in $\hyrql$, whereas $\texttt{FOQ}$ is restricted to lists
    of qubits. 

  \item \textit{General recursion} specifies whether the language includes
    unbounded recursion on classical data. The ability to handle this type of recursion is
    necessary for a resource analysis to be relevant (\ie not trivial as in
    the case of strong normalizing languages). 

  \item \textit{Higher-order} denotes the ability of the language to
    feature general higher-order terms, i.e., functions taking functions as
    input. While some languages of Table~\ref{tab:related} have syntactic
    constructs to do so, most of them are restricted to first-order calls. 
\end{itemize}

As we are interested in resource analysis in the general setting,
Table~\ref{tab:related} only compares \hyrql{} with hybrid languages, rather
than with existing languages with quantum control, non-exhaustively
\cite{SVV18,DCM22,Yin24,YVC24}.
Other studies have already been carried out on resource analysis in quantum
programming languages. They use techniques (\eg type systems) developed in the
field of Implicit Computational Complexity (ICC, see~\cite{Pec20}) to
characterize quantum complexity classes~\cite{DLMZ10}, to infer the expected
cost or expected value of a quantum program~\cite{AMPPZ22}, or to infer upper
bounds on quantum resources (depth and size of circuits)~\cite{CDL25}. However,
these studies are restricted to classical control. For example, \cite{DLMZ10}
studies the use of soft linear logic~\cite{L04} on a quantum lambda
calculus~\cite{SV05}, \cite{AMPPZ22} adapts expectation transformers on a
language based on $\texttt{QPL}$~\cite{Sel04}, and~\cite{CDL25} develops a
dependent type system on a variant of the Quipper circuit description
language~\cite{GLRSV13}.

\section{The $\hyrql$ Programming Language}
\label{sec:language}
We introduce the syntax, the operational semantics, and the type system of
\hyrql{} along with illustrating examples.

\subsection{Syntax}

\begingroup
\renewcommand{\arraystretch}{1.25}
\begin{table}[t]
  \centering
  \[
    \begin{array}{>{$}r<{$} @{\quad} r @{~\Coloneqq~} l}
      (Values) & v & x \ \, \mid\, \zket \gmid \oket \gmid
      c(v_1, \dots, v_n) \gmid \lbd x t \gmid \letrec f x t
      \gmid \sum_{i=1}^n \alpha_i \cdot v_i\\
      (Terms) & t & x\;
      \begin{array}[t]{l}
        \mid\, \zket \gmid \oket \gmid \qcase{t}{t_0}{t_1} \\
        \mid\, c(t_1, \dots, t_n) \gmid \matchdefault t \\
        \mid\, \lbd x t \gmid \letrec f x t \gmid t_1 t_2
        \gmid \sum_{i=1}^n \alpha_i \cdot t_i \gmid \shape t
      \end{array}
    \end{array}
  \]
  \caption{Syntax of the $\hyrql$ language}
  \label{tab:syntax}
\end{table}
\endgroup

The syntax of Terms and Values in $\hyrql$ is provided in
Table~\ref{tab:syntax}. Terms feature variables, noted $x, y, f, \dots$ and
taken from an infinite countable set of variables. We denote by $\seq{e}$ a
(finite and possible empty) sequence of elements $e_1, \dots, e_n$. Qubits are
introduced through basis states $\zket, \oket$, allowing us to write quantum
conditionals $\qcasedefault t$, which corresponds to the superposition of $t_0$
and $t_1$ controlled by qubit $t$. $\hyrql$ also features classical constructors
$c^s$ consisting of a constructor symbol $c$, and a signature $s$ indicating
the type of the constructor. This will be formalized in
Section~\ref{subsec:types} and will be used only when required.
Constructors can be used in constructor application $c(t_1, \dots, t_n)$
and through pattern-matching, sometimes abbreviated as $\matchdefaultcollapsed
t$. Each constructor symbol $c$ comes with a fixed arity and is always fully
applied. We assume the existence of standard inductive constructors for unit
$()$, tensor products $\otimes$, natural numbers $0$ and $S$, and lists $\nil$
and $::$. For convenience, constructors are sometimes used in an infix
notation. For example, tensor products will be written as $x \otimes y$ and
list constructors as $h::t$.

Higher-order is featured via a standard $\lambda$-abstraction $\lbd x t$, and a
construct for general recursion $\letrec f x t$. Term application is denoted by
$t_1t_2$.

Given amplitudes $\alpha_i \in \mathbb C$, the term $\sum_{i=1}^n \alpha_i
\cdot t_i$ represents a \emph{superposition} of terms $t_i$. In the special
case where $n=2$, we just write $\alpha_1 \cdot t_1+\alpha_2 \cdot t_2$. In a
dual manner, the $\textshape$ construct, introduced in~\cite{GLRSV13},
returns the classic structural information on data containing quantum
information. For example, the shape of a qubit list is a copy of the list
structure without the qubits (see Example~\ref{ex:shape-semantics}).

In order to avoid conflicts
between free and bound variables, we will always work up to $\alpha$-conversion
and use Barendregt's convention~\cite[p. 26]{Bar84} which consists in keeping
all bound and free variable names distinct, even when this remains implicit.
We also define a partial order $\subterm$ on terms, where $s \subterm t$ if $s$
is a \emph{strict subterm} of $t$. We denote $\subtermeq$ as the reflexive
closure of $\subterm$. The \emph{size of a term} is written $\size t$ and is
defined standardly in Definition~\ref{def:size}.

As we work with amplitudes and superpositions, Terms and Values have to be
considered with respect to a vector space
structure. Towards that end, we define an equivalence relation $\equiv$ in
Table~\ref{tab:equiv}. Such definition relies on \emph{equivalence contexts},
defined by the following grammar:
\[
  \begin{aligned}
    (\text{Equiv. contexts}) \quad   \equivcontext \Coloneqq \diamond & \gmid
    \qcasedefault \equivcontext \gmid c(\seq{t_1}, \equivcontext, \seq{t_2})
    \gmid t \equivcontext \gmid \equivcontext t\\
    &\gmid \matchdefaultcollapsed{\equivcontext}  \gmid t + \equivcontext \gmid
    \shape \equivcontext
  \end{aligned}
\]
An equivalence context $\equivcontext$ is thus a term with one hole $\diamond$;
let $\equivcontext[t]$ be the term obtained by filling the hole with
$t$ in $\equivcontext$.
The first rules of Table~\ref{tab:equiv} make the definition of the $\sum$
construct unambiguous/sound, i.e., it corresponds exactly to repeated
applications of the $+$ construct. The last rules of Table~\ref{tab:equiv} are
commutation rules highlighting the linearity of quantum terms.

\begingroup
\renewcommand\displaystyle\textstyle
\newcommand\sep{\\[1.25ex]}
\begin{table}[t]
  \centering
  \[
    \begin{array}{c}
      t_1 + t_2 \equiv t_2 + t_1
      \qquad
      t_1 + (t_2 + t_3) \equiv (t_1 + t_2) + t_3
      \qquad
      1 \cdot t \equiv t
      \qquad
      t + 0 \cdot t' \equiv t
      \sep
      \alpha \cdot (\beta \cdot t) \equiv (\alpha \beta) \cdot t
      \qquad
      \alpha \cdot (t_1 + t_2) \equiv \alpha \cdot t_1 + \alpha \cdot t_2
      \qquad
      \alpha \cdot t + \beta \cdot t \equiv (\alpha + \beta) \cdot t
      \\[1.2em]
      \begin{aligned}
        \qcasedefault {(\sum_{j=1}^m \alpha_j \cdot s_j)} &\equiv
        \sum_{j=1}^m \alpha_j \cdot (\qcasedefault{s_j})
        \sep
        c(\seq{t_1}, \sum_{j=1}^m \alpha_j \cdot s_j,
        \seq{t_2}) &\equiv \sum_{j=1}^m \alpha_j \cdot
        c(\seq{t_1}, s_j, \seq{t_2})
        \sep
        \matchdefaultcollapsed{(\sum_{j=1}^m \alpha_j \cdot s_j)}
        &\equiv \sum_{j=1}^m \alpha_j \cdot (\matchdefaultcollapsed{s_j})
        \sep
        t(\sum_{i=1}^m \alpha_j \cdot s_j) &\equiv \sum_{j=1}^m
        \alpha_j \cdot t\,s_j
        \sep
        \equivcontext[t] &\equiv \equivcontext[t'] \qquad \text{ when
        } t \equiv t'
      \end{aligned}
    \end{array}
  \]
  \caption{Equivalence relation $\equiv\ \subseteq\ \text{Terms}
  \times \text{Terms}$}
  \label{tab:equiv}
\end{table}

\endgroup

\begin{example}[Quantum State and Hadamard Gate]
  \label{ex:had-syntax}
  We can define the orthogonal basis states as $\ket \pm \triangleq
  \frac{1}{\sqrt 2} \cdot \zket \pm \frac{1}{\sqrt 2} \cdot \oket$. The
  Hadamard gate can be encoded in \hyrql{}:
  \[
    \mathtt{Had} \triangleq \lbd x \qcase x {\ket +} {\ket -}
  \]
\end{example}

\begin{example}[Quantum Switch]
  \label{ex:qs-syntax}
  The quantum switch~\cite{CAPV13} is a program which, on two unitaries $U, V$
  and a two-qubit state $\ket{\phi}$, applies $UV$ and $VU$, in superposition,
  to the second qubit, controlled by the value of the first qubit. This can be
  implemented in \hyrql{} by $\QS\ U\ V \ket{\phi}$, where $\QS$ is defined as:
  \[
    \QS{} \triangleq\lbd f\lbd g\lbd q \match q { c \otimes t \to
      \qcase{c}{\zket \otimes f(g\,t)}{ \oket \otimes g(f\,t)}
    }
  \]
\end{example}

\begin{example}[Bounded quantum walk]
  \label{ex:bqwalk}
  \hyrql{} allows one to write recursive programs that are classically and
  quantumly controlled. The program $\mathtt{bqwalk}$ produces all the possible
  quantum walk paths of size at most $n$, starting from a quantum state $q$:
  \[
    \begin{aligned}
      \mathtt{repeat} &\triangleq \letrec g n \match n {
        0 \to \nil,
        S(m) \to \zket :: g\,m
      } \\
      \mathtt{bqwalk} &\triangleq \letrec f q \lbd n \qcasesplit q {
        \zket::(\mathtt{repeat}~n)
      } {
        \oket :: 
        \match n {
          0 &\to \nil \\
          S(m) &\to (f\,(\mathtt{Had}\,\oket)\,m)
        }
      }
    \end{aligned}
  \]
\end{example}

The programs from Examples~\ref{ex:qs-syntax} and \ref{ex:bqwalk} are strictly
more than quantum circuits, as one execution of the program does not provide a
full description of the corresponding behaviour. Thus, \hyrql{} is more than a
circuit description language~\cite{GLRSV13}. We also provide the hybrid quantum
protocol BB84~\cite{BB84} and the Quantum Fourier Transform in
Appendix~\ref{app:additional}.

\subsection{Operational Semantics}
\label{subsec:semantics}

\begingroup
\newcommand{\ptlabel}[1]{
  \begin{prooftree}[center=false, right label template =
      \footnotesize\inserttext,
    template = \footnotesize$\inserttext$] #1
\end{prooftree}}
\newcommand{\ptnolabel}[1]{
  \begin{prooftree}[center=false, right label template = ,
    template = \small$\inserttext$] #1
\end{prooftree}}
\renewcommand{\qquad}{\quad\,\,\,}
\renewcommand{\arraystretch}{2.5}
  \begin{table}[t]
    \[
      \begin{array}{c}
        \ptlabel{
          \infer0[(Qcase$_0$)]{\qcase {\zket} {t_0} {t_1 }\reduces t_0}
        }
        \qquad
        \ptlabel{
          \infer0[(Qcase$_1$)]{\qcase {\oket} {t_0} {t_1 }\reduces t_1}
      }
      \\
      \ptlabel{
        \infer0[(Match)]{\matchdefaultcollapsed{c_j(\seq{v_p})
        } \reduces t_j\{ \seq{v_p}/\seq{x_j} \}}
      }
      \qquad
      \ptlabel{
        \infer0[(Lbd)]{(\lbd x t) v_p \reduces t \{v_p/x\}}
      }
      \\
      \ptlabel{
        \infer0[(Rec)]{
          (\letrec f x t)v_p \reduces t \{\letrec f x t/f, v_p/x\}
        }
      }
      \qquad
      \ptlabel{
        \hypo{
          \sum_{i=1}^n \alpha_i \cdot p_i \in \can
          \setminus \valueset
        }
        \hypo{p_i \zoreduces t_i}
        \infer2[(Can)]{\sum_{i=1}^n \alpha_i \cdot p_i \reduces
        \sum_{i=1}^n \alpha_i \cdot t_i}
      }
      \\
      \ptlabel{
        \infer0[(Shape$_0$)]{\shape {\ket 0} \reduces \vunit}
      }
      \qquad
      \ptlabel{
        \infer0[(Shape$_1$)]{\shape {\ket 1} \reduces \vunit}
      }
      \qquad
      \ptlabel{
        \infer0[(Shape$_c$)]{\shape {c^s(\seq{v_p})}
        \reduces   c^{\shape{s}}(\shape{\seq{v_p}})}
      }
      \\
      \ptlabel{
        \hypo{\sum_{i=1}^n \alpha_i \cdot v_i \in \can}
        \infer1[(Shape$_s$)]{\shape{\sum_{i=1}^n \alpha_i \cdot
        v_i} \reduces \shape{v_1}}
      }
      \qquad
      \ptlabel{
        \hypo{t \reduces t'}
        \infer1[(Shape$_\evalcontext$)]{\shape t \reduces \shape{t'}}
      }
      \\
      \ptlabel{
        \hypo{p \reduces t}
        \infer1[($\evalcontext$)]{\evalcontext[p] \reduces \evalcontext[t]}
      }
      \qquad
      \ptlabel{
        \hypo{t \equiv t_1}
        \hypo{t_1 \reduces t_1'}
        \hypo{t_1' \equiv t'}
        \infer3[(Equiv)]{t \reduces t'}
      }
  \end{array}\]
  \caption{Reduction rules of the language}
  \label{tab:reduction}
\end{table}
\endgroup

In this section, we define the call-by-value operational semantics of $\hyrql$.
As we are interested in resource analysis, we want to avoid reducing
non-meaningful terms, e.g., a reduction of $t$ in $s + 0 \cdot t \ \equiv\ s$
is meaningless. Towards that end, we introduce \emph{canonical forms}
(Definition~\ref{def:canonical}), which take care of removing ill-behaved terms
before applying a reduction. It relies on \emph{pure terms}, which are defined
by the following grammar:
\[
  \begin{aligned}
    (\text{Pure terms}) \quad  p  \Coloneqq x &\gmid \zket \gmid \oket
    \gmid \qcasedefault p  \gmid c(p_1, \dots, p_n) \\
    &\gmid \matchdefaultcollapsed p \gmid \lbd x t \gmid \letrec f x t
    \gmid tp \gmid \shape t
  \end{aligned}
\]
A \emph{pure value}, denoted $v_p$, is a pure term that is a value.
\begin{definition}[Canonical form]
  \label{def:canonical}
  A canonical form of a term $t$ is any term $t'$ such that
  $t\equiv t'$ and $t'=\sum_{i=1}^n \alpha_i \cdot p_i $, where $p_i$ are
  pure terms, $\alpha_i \neq 0$ and $p_i \equiv p_j$ implies  $i = j$. The
  set of canonical forms is denoted by $\can$.
\end{definition}

The canonical form always exists and is unique for well-typed terms
(Lemma~\ref{lem:canonical-typed}). Note that the naming \emph{pure terms} comes
from~\cite{DCM22} and is not related to \emph{pure states} in quantum computing,
as pure terms could semantically yield a mixed quantum state.

We also define \emph{evaluation contexts} by the following grammar:
\begin{align*}
  (\text{Evaluation contexts}) \quad \evalcontext \Coloneqq \, \diamond
  &\gmid \qcase{\evalcontext}{t_0}{t_1} \gmid c(\seq{p},
  \evalcontext, \seq{v_p}) \\
  &\gmid \matchdefaultcollapsed{\evalcontext} \gmid t \evalcontext \gmid
  \evalcontext v_p
\end{align*}
Again, let  $\evalcontext[t]$ be the term obtained by filling the
hole $\diamond$ with $t$ in $\evalcontext$.

The operational semantics of $\hyrql$ is described in Table~\ref{tab:reduction}
as a reduction relation $\reduces \ \subseteq
\ \text{Terms} \times \text{Terms}$, where $\{t/x\}$ denotes the
standard substitution of a variable $x$ by a term $t$. The reduction
implements a \emph{call-by-value} strategy.

In the rule (Can), $t\zoreduces t'$  holds if either $t\reduces t'$
or, $\neg(\exists t'',\ t \reduces t'')$ and $t' = t$. Intuitively,
it means that we apply one reduction for each element of a
superposition, when possible.
In the rule (Shape$_c$), $\shape{\seq v}$ is syntactic sugar for $\shape{v_1},
\dots, \shape{v_n}$, and the shape of a signature $\shape s$ is defined
formally in Definition~\ref{def:shape-type}.
In all rules of Table~\ref{tab:reduction}, except for the (Can) and
(Equiv) rules, the left-hand-side term is a pure term (i.e., not a
superposition). This implies that any summation must be
expressed as a canonical form, which avoids reducing
a subterm with a null amplitude, or more generally to reduce two
identical terms that would sum up to $0$.  The goal is to avoid
reductions that have no physical meaning, as this would invalidate
all resource analysis results.

We define $\reduces^*$ as the reflexive and transitive closure of
$\reduces$.
For $k \in \mathbb{N}$, we also write $t \mostreduces k t'$, when $t$
reduces to $t'$ in at most $k$ steps.
A term $t$ \emph{terminates}, if any chain of reduction starting
from $t$ reaches a value, meaning that $t \mostreduces k v$ holds for some $k$.

\begin{example}
  \label{ex:had-semantics}
  Let us consider the $\mathtt{Had}$ program from Example~\ref{ex:had-syntax},
  we can check that it gives the desired result when applying it to $\zket$:
  \[
    \begin{aligned}
      \mathtt{Had} \zket &\reduces \qcase \zket {\ket +}{\ket -}
      &&\text{via (Lbd)} \\
      &\reduces \ket + &&\text{via (Qcase$_0$)}
    \end{aligned}
  \]
  Since the Hadamard gate is its own inverse we can recover $\ket0$.
  Given the following term  $h_v= \qcase{v} {\ket+}{\ket -}$,
  we have the following:
  \[
    \mathtt{Had} \ket + \equiv \frac 1 {\sqrt 2} \cdot \mathtt{Had} \zket +
    \frac 1 {\sqrt 2} \cdot \mathtt{Had} \oket
    \reduces \frac{1}{\sqrt 2} \cdot h_0
    + \frac{1}{\sqrt 2} \cdot h_1 \reduces \frac{1}{\sqrt 2} \cdot \ket + +
    \frac{1}{\sqrt 2} \cdot \ket - \equiv \ket 0
  \]
  Note that the reductions are performed in parallel on each term of the
  superpositions using rule (Can) of Table~\ref{tab:reduction}.
 \end{example}

\begin{example}
  \label{ex:shape-semantics}
  Given a qubit list $l = \zket :: \oket :: \ket + :: \nil$, one can verify that
  $\shape l \sreduces \vunit :: \vunit :: \vunit :: \nil$ (note that $::$ and
  $[\ ]$ use distinct signatures in $l$ and the reduct of $\shape l $).  The
  result is a list of same length as $l$, with no quantum data anymore, hence
  it can be treated non-linearly. This operation is independent of the value of
  each qubit in the list.
\end{example}

It is worth mentioning that while we can write qubit streams in \hyrql{}, the
call-by-value strategy will require it to be evaluated before doing any
operation on such data.

\subsection{Type System}
\label{subsec:types}

\textbf{Types and Contexts. }
Types in \hyrql{} are provided by the following grammar:
\[
  \text{(Types)} \qquad T \Coloneqq \qbit \gmid B \gmid T \linfunc T
  \gmid T \nonlinfunc T
\]
$\qbit$ is the type for qubits and $\bvar$ is a \emph{constructor
type} from a fixed set $\bset$.
A \emph{basic type}, noted $\bkvar$, is either a qubit type, or a
constructor type $\bvar$, i.e., $\bkvar \in \{\qbit\} \cup \bset$.
Each constructor type is defined by its set of constructors:
\[
  \cons B \triangleq \{c^s \gmid s = \signdefault\}
\]

A constructor symbol $c$ comes with a signature $s$ of arity $n$, defining the
inputs and output types of $c$. We may write $c^\bvar$ when $s$ is of arity
$0$. By construction, $\cons B \cap \cons{B'} = \emptyset$ when $B \neq B'$.
In what follows, we will consider the
following constructor types: a unit type $\unit$ with a unique constructor
$\vunit^\unit$; tensor types $\bkvar_1 \otimes \bkvar_2$ for given
types $\bkvar_1, \bkvar_2$ with a unique constructor
$\otimes^{\sign{\bkvar_1, \bkvar_2}{\bkvar_1 \otimes \bkvar_2}}$;
natural numbers $\nat$ with constructors $0^\nat$ and 
$S^{\sign \nat \nat}$; lists $\typelist \bkvar$ for a given type $\bkvar$,
with constructors $\nil^{\typelist \bkvar}$ and $::^{\sign {\bkvar, \typelist
\bkvar} {\typelist \bkvar}}$. We consider that $\bset$ always contains the unit
type $\unit$.

The language also features two distinct \emph{higher-order} types
for linear and non-linear functions, denoted respectively
by $\linfunc$ and $\nonlinfunc$.

In the typing discipline, it will be useful to distinguish between
types based on whether or not they contain quantum data. The latter
must follow the laws of quantum mechanics (no-cloning), whereas
classical types are more permissive.
Towards that end, we define inductively the set of \emph{quantum constructor
types} $\bqset$ by 
\[
  \bqset \triangleq \{\bvar \in \bset \gmid
    \exists\,c^{\signdefault}, \exists\,1 \leq
i \leq n, \bkvar_i \in \bqset \cup \{ \qbit\}\}. 
\]
A quantum
constructor type $\bvar_\qvar \in \bqset$ has (at least) a
constructor symbol with a type $\qbit$ or (inductively) a quantum
constructor type in its signature.
The set $\bcset$ of \emph{classical constructor types} $\bvar_\cvar$
is defined by $\bcset \triangleq \bset \setminus \bqset$.
For example, $\unit,\ \nat \in \bcset$ whereas $\typelist{\qbit} \in \bqset$.
This allows us to split the types between \emph{quantum types} $\qvar$ and
\emph{classical types} $\cvar$:
\[
  \qvar \ \Coloneqq \ \qbit \gmid \bqvar
  \qquad
  \cvar \ \Coloneqq \ \bcvar \gmid T \linfunc T \gmid T \nonlinfunc T
\]
It is now assumed that the inputs of each constructor are ordered, containing
first variables of classical type, then variables of quantum type, enabling a
hybrid behaviour.

\emph{Typing contexts} are defined as follows:
\[
  \Gamma, \Delta \quad \Coloneqq \quad \varnothing \gmid \set{x : T}
  \gmid \set{\shapemarker{x : T}} \gmid \Gamma \cup \Delta,
\]
where $\shapemarker{x : T}$ is called a \emph{boxed variable},
indicating that $x$ is a linear variable captured by a \textshape{} construct.
We define the \emph{domain} of a context $\Delta$ as the set of its variables,
i.e., $\typedom \Delta \triangleq \set{x \gmid \exists\,T,\,x : T
\in \Delta \vee \shapemarker{x : T} \in \Delta}$.
Whenever we write $\Gamma, \Delta$ or
$\Gamma;\Delta$, it is assumed that $\Gamma$ and $\Delta$ are compatible, that
is, $\typedom \Gamma \cap \typedom \Delta = \emptyset$.
We use the shorthand notation $\vartyping {\seq x} {\seq T}$ for $\vartyping
{x_1}{T_1}, \dots, \vartyping {x_n}{T_n}$; and $\shapemarker \Delta$ for
$\shapemarker{x_1 : T_1}, \dots, \shapemarker{x_n : T_n}$.

\medskip
\noindent
\textbf{Typing rules. }
A \emph{typing judgment} is written $\typing \Gamma \Delta t T$, describing
that the term $t$ is \emph{well-typed}, with type $T$, under \emph{non-linear
context} $\Gamma$ and \emph{linear context} $\Delta$.
The typing rules of the language are defined in Table~\ref{tab:typing}.
A well-typed term $\typing \Gamma \Delta t T$ is \emph{closed} if $\Gamma =
\Delta = \emptyset$.
A well-typed closed term $t$ is said to \emph{terminate over any input}, if for
any well-typed closed value $v_1, \dots, v_n$ terminating over any input such
that $tv_1\dots v_n$ is a well-typed closed term of type $\bkvar$, $tv_1 \dots
v_n$ terminates.

\begingroup
\newcommand{\ptlabel}[1]{
  \begin{prooftree}[center=false, right label template = \small\inserttext,
    template = \small$\inserttext$] #1
\end{prooftree}}
\newcommand{\ptnolabel}[1]{
  \begin{prooftree}[center=false, right label template = ,
    template = \small$\inserttext$] #1
\end{prooftree}}
\renewcommand{\arraystretch}{3}
\renewcommand\pt\ptlabel
\begin{table}[t]
  \[
    \begin{array}{c}
      \pt{
        \infer0[(ax)]{\typing \Gamma {\vartyping{x}{T}} x T}
      }
      \qquad
      \pt{
        \infer0[(ax$_c$)]{\typing {\Gamma, \vartyping x C} {} x C}
      }
      \qquad
      \pt{
        \infer0[(ax$_0$)]{\typing \Gamma {} {\zket} \qbit}
      }
      \\
      \pt{
        \infer0[(ax$_1$)]{\typing \Gamma {} {\oket} \qbit}
      }
      \quad
      \ptlabel{
        \hypo{\typing \Gamma \Delta t \qbit}
        \hypo{\typing \Gamma {\Delta'} {t_0} Q}
        \hypo{\typing \Gamma {\Delta'} {t_1} Q}
        \hypo{t_0 \perp t_1}
        \infer4[(qcase)]{\typing \Gamma {\Delta, \Delta'} {\qcase
        t {t_0} {t_1}} Q}
      }
      \\
      \pt{
        \hypo{s = \signdefault}
        \hypo{\typing \Gamma {\Delta_i} {t_i} {\bkvar_i}}
        \infer2[(cons)]{\typing \Gamma {\Delta_1, \dots, \Delta_n}
        {c^s(t_1, \dots, t_n)} \bvar}
      }
      \\
      \ptlabel{
        \hypo{\cons \bvar = \set{c_i^{s_i}}_{i=1}^n}
        \hypo{s_i = \sign{\seq{C_i}, \seq{Q_i}} \bvar}
        \hypo{\typing {\Gamma_1}{\Delta_1} t \bvar}
        \hypo{\typing
          {\Gamma_2, \vartyping{\seq{y_i}}{\seq{C_i}}}
          {\Delta_2, \vartyping{\seq{z_i}}{\seq{Q_i}}}
          {t_i} \bkvar
        }
        \infer4[(match)]{
          \typing {\Gamma_1, \Gamma_2} {\Delta_1, \Delta_2} {
            \textmatch_{1 \leq i \leq n}\,t\,\{
              c_i^{s_i}(\seq{y_i}, \seq{z_i}) \to t_i
            \}
          } \bkvar
        }
      }
      \\
      \pt{
        \hypo{\typing \Gamma {\Delta, \vartyping x T} t T'}
        \infer1[(abs)]{\typing \Gamma \Delta {\lbd x t} {T \linfunc T}'}
      }
      \qquad
      \pt{
        \hypo{\typing {\Gamma, \vartyping x C} \Delta t T}
        \infer1[(abs$_c$)]{\typing \Gamma \Delta {\lbd x t} {C \nonlinfunc T}}
      }
      \qquad
      \pt{
        \hypo{\typing {\Gamma, \vartyping f T} {} {\lbd x t} T}
        \infer1[(rec)]{\typing \Gamma {} {\letrec f x t} T}
      }
      \\
      \qquad
      \pt{
        \hypo{\typing \Gamma \Delta {t_1} {T \linfunc T'}}
        \hypo{\typing \Gamma {\Delta'} {t_2} T}
        \infer2[(app)]{\typing \Gamma {\Delta, \Delta'} {t_1t_2} {T'}}
      }
      \qquad
      \pt{
        \hypo{\typing \Gamma \Delta {t_1} {C \nonlinfunc T}}
        \hypo{\typing \Gamma {} {t_2} C}
        \infer2[(app$_c$)]{\typing \Gamma \Delta {t_1t_2} T}
      }
      \\
      \hypertarget{typesup}{
      \ptlabel{
        \hypo{\typing \Gamma \Delta {t_i} \qvar}
        \hypo{\sum_{i=1}^n \size{\alpha_i}^2 = 1}
        \hypo{\forall i \neq j,\ t_i \perp t_j}
        \infer3[(sup)]{\typing \Gamma \Delta{ \sum_{i=1}^n \alpha_i
        \cdot t_i} \qvar}
      }
      }
      \quad
      \ptlabel{
        \hypo{\typing \Gamma \Delta t \bkvar}
        \infer1[(shape)]{\typing {\Gamma, \shapemarker \Delta} {}
        {\shape{t}} {\shape \bkvar}}
      }
      \\
      \ptlabel{
        \hypo{\typing {\Gamma,\shapemarker{\vartyping x \bkvar}}
        {\Delta, \vartyping y \bkvar} t T}
        \infer1[(contr)]{\typing \Gamma {\Delta, \vartyping y \bkvar}
        {t\{y/x\}} T}
      }
      \qquad
      \ptlabel{
        \hypo{\typing \Gamma \Delta t T}
        \hypo{t \equiv t'}
        \infer2[(equiv)]{\typing \Gamma \Delta {t'} T}
      }
  \end{array}\]
  \caption{Typing rules of the language}
  \label{tab:typing}
\end{table}
\endgroup

Typing rules rely on an external \emph{orthogonality} predicate on terms. This
uses the notion of $\Delta$-\emph{context substitutions}, which, for a given
$\Delta$, denote the substitutions $\sigma$ such that for any $x : T$ with
$\vartyping x T \in \Delta$ or $\shapemarker{\vartyping x T} \in \Delta$,
$ \sigma(x)=v$, where $v$ is a well-typed closed value of type $T$
that terminates over any given input. It also requires the definition of an
an \emph{inner product} on values with a canonical form, defined
below, where $\kron x y$ is the Kronecker symbol.
\[
  \langle v, w \rangle \triangleq \sum_{i=1}^n \sum_{j=1}^m \alpha_i \beta_j^*
  \kron{v_i}{w_j}, \quad \text{ where }
  v \equiv \sum_{i=1}^n \alpha_i \cdot v_i \in \can \text{ and }
  w \equiv \sum_{j=1}^m \beta_j \cdot w_j \in \can
\]
\begin{definition}[Orthogonality]
  \label{def:orthogonality}
  Let $\typing \Gamma \Delta t Q$ and $\typing \Gamma \Delta {t'} Q$ be
  two well-typed terms. $t$ and $t'$ are \emph{orthogonal}, written $t \perp
  t'$, if for any $\Gamma \cup \Delta$-context substitution $\sigma$:
  \begin{itemize}
    \item $ \sigma(\shape t) \sreduces v$, $\sigma(\shape{t'})
      \sreduces v'$ and $v = v'$;
    \item $\sigma(t) \sreduces v \in \can$,
      $\sigma(t') \sreduces v' \in \can$,
      and $\langle v, v' \rangle = 0$.
  \end{itemize}
\end{definition}

These conditions imply that the reduced values share the same classical
structure, and that their quantum data is orthogonal. For example, orthogonal
lists must be of the same size. This is used in the typing rules
(qcase) and (sup) of Table~\ref{tab:typing} to ensure the feasibility of typed terms, similarly to
existing languages \cite{DCM22,Lem24,DLPZ25}. Note that the rules requiring
this predicate apply it only to well-typed terms, thus the type system is not
cyclic. While an alternative condition could have been given with dependent
types, this gives a uniform process that is defined for any type, including
types that could be defined later by a programmer, and a simpler type system.
As we will prove that values have a unique canonical form
(Lemma~\ref{lem:canonical-typed}), the inner-product condition is sound.

We also define the \emph{shape of a basic type} below.
\begin{definition}
  \label{def:shape-type}
  The shape of a basic type is defined inductively as $\shape \qbit \triangleq
  \unit$, and $\shape \bvar $ is the type defined by the constructor set  
  \[
    \cons{\shape \bvar} \triangleq \set{
      c^{\shape{s}}
      \gmid c^{s} \in \cons \bvar
    }
  \]
  where the shape of a signature $s=\bkvar_1,\ldots,\bkvar_n \to \bvar$ is defined as 
  \[
  \shape{s} \triangleq \shape{\bkvar_1},\ldots,\shape{\bkvar_n} \to \shape{\bvar}.
  \]
\end{definition}
In particular, the shape of a list of qubits is a list of units, i.e.,
$\shape{\typelist{\qbit}} = \typelist{\unit}$, which is coherent with the
reduction rules in Table~\ref{tab:reduction}. By construction, it always holds
that $\shape \bkvar \in \bcset$. This is used in the typing rule (shape) of
Table~\ref{tab:typing}, which extracts the classical structure of a quantum
term, boxes any variable in $\Delta$, and puts it in the non-linear context.
The judgment obtained has an empty linear context and thus can be used
classically (i.e., duplicated or discarded). A boxed variable $x$ can be
removed using rule (contr) of Table~\ref{tab:typing}, by associating with a
linear variable of same type $y$, allowing to read both the quantum data and
the classic structural information of $y$ (see Example~\ref{ex:len-typing}).

\medskip
\noindent
\textbf{Typing derivations. }
We write $\pi \triangleright \typing \Gamma \Delta t T$ for the tree of root
$\typing \Gamma \Delta t T$ derived using the rules of Table~\ref{tab:typing}.
We also denote by $\pi \rightslice \mathrm{(r)}[t_1, \dots, t_n]$ the subtree
of $\pi$ rooted by (r) and whose premises are $t_1, \dots, t_n$.

\begin{example}\label{ex:len-typing}
  The length of a list can be computed by the following term:
  \[
    \mathtt{len} \triangleq \letrec fx \match x {
      \nil \to 0, h::t \to S(f\,t)
    }
  \]
  Notice that we discard $h$, thus the list is not used linearly. Therefore, we
  can only type this as a non-linear function, i.e., $\closedtyping
  {\mathtt{len}} \typelist B \nonlinfunc \nat$ for a given $B \in \bcset$.
  However, we can still compute the length of a qubit list with $\lbd x (x
  \otimes \mathtt{len}(\shape x))$, whose typing derivation $\pi$ is written
  below. For conciseness, we hide the signature of $\otimes$.
  \[
    \begin{prooftree}[
        center=false, right label template = \footnotesize\inserttext, rule
        margin = 1.2ex, template = \footnotesize$\inserttext$
      ]
      \infer0[(ax)]{
        \typing{\shapemarker{y : \typelist \qbit}}{x : \typelist
        \qbit}{x}{\typelist \qbit}
      }
      \infer0[(ax)]{\typing{}{y : \typelist \qbit}{y}{\typelist \qbit}}
      \infer1[(shape)]{
        \typing{\shapemarker{y : \typelist \qbit}}{}{\color{red}\shape y}{\typelist \unit}
      }
      \infer1[\color{red}(app$_c$)]{
        \typing{\shapemarker{y : \typelist
        \qbit}}{}{\mathtt{len}(\shape y)}{\nat}
      }
      \infer2[(cons)]{
        \typing{\shapemarker{y : \typelist \qbit}}{x : \typelist \qbit}{x
        \otimes \mathtt{len}(\shape y)}{\typelist \qbit \otimes \nat}
      }
      \infer1[(contr)]{
        \typing{}{x : \typelist \qbit}{x \otimes \mathtt{len}(\shape
        x)}{\typelist \qbit \otimes \nat}
      }
      \infer1[(app)]{
        \typing{}{}{\lbd x (x \otimes \mathtt{len}(\shape
        x))}{\typelist \qbit \linfunc \typelist \qbit \otimes \nat}
      }
    \end{prooftree}
  \]
  As $\shape y$ is of classical constructor type, it can be used as an input of
  $\mathtt{len}$, allowing us to compute the length of $y$, given by the
  subtree $\pi \rightslice {\color{red}(\mathrm{app}_c) [\shape y]}$. The rule
  (contr) removes the marker introduced by (shape), by substituting $y$ with a
  variable of same type, $x$, thus yielding a closed term. Note that while $x$
  appears twice in the final term, the quantum information of $x$ has a linear
  behaviour, and only the classical structure is used non-linearly.
\end{example}

\section{Main Results}
\label{sec:results}
This section is devoted to checking that {\hyrql} satisfies usual properties of
typed languages. Moreover, it studies the complexity of checking the
orthogonality predicate. We also prove that we can compile an expressive subset
of terms of \hyrql{} into quantum circuits, with a guarantee on the size of the
circuit. This gives us a characterization of the class of functions computable
in quantum polynomial time, known as \fbqp{}~\cite{BV97}.

\subsection{Standard Properties}
\label{subsec:properties}

By design of the type system, the non-linear context $\Gamma$ can be weakened
and typing still holds. This result does not hold for $\Delta$, as intended,
due to its linear behavior.

\begin{restatable}[Weakening]{proposition}{weakening}
  \label{lem:weakening}
  Let $\typing \Gamma \Delta t T$ be a well-typed term.
  Then $\typing {\Gamma, \Gamma'} \Delta t T$ holds for any context
  $\Gamma'$ compatible with $\Gamma, \Delta$.
\end{restatable}

The reduction relation $\reduces$ is confluent up to equivalence for
well-typed terms.

\begin{restatable}[Confluence]{theorem}{confluence}
  \label{thm:confluence}
  Given a well-typed term $t$, if there exist $t_1$ and $t_2$ such
  that $t \sreduces t_1$ and $t\sreduces t_2$,
  then there exist $t_3$ and $t_4$
  such that $t_1 \sreduces t_3$, $t_2 \sreduces t_4$ and $t_3 \equiv t_4$.
\end{restatable}
Due to rule (Equiv) of Table~\ref{tab:reduction}, confluence is up to
equivalence $\equiv$. Moreover, confluence only holds on well-typed terms.
A consequence of Theorem~\ref{thm:confluence} is that any terminating term has
a unique normal form, up to equivalence. This will be useful for progress
(Theorem~\ref{lem:progress}).

\begin{restatable}[Canonical form for typed terms]{lemma}{canonicaltyped}
  \label{lem:canonical-typed}
  Let $\typing \Gamma \Delta t T$ be a well-typed term.
  Then $t$ has a canonical form $\sum_{i=1}^n \alpha_i \cdot p_i$,
  and this canonical form is unique, up to equivalence on each $p_i$
  and term reordering. Furthermore, $\forall\ 1 \leq i \leq n,\ \typing \Gamma
  \Delta {p_i} T$.
\end{restatable}

As there are no distinct and equivalent pure values,
Lemma~\ref{lem:canonical-typed} implies that the canonical forms of
Definition~\ref{def:orthogonality} exist and are unique.

Typing implies that the values are the normal forms
of the language, up to equivalence.

\begin{restatable}[Progress]{theorem}{progress}
  \label{lem:progress}
  Let $\closedtyping t T$ be a closed term, either $t$ is equivalent
  to a value, or $t$ reduces.
\end{restatable}

Contrary to standard progress lemmas, we require $t$ to be a value up to
equivalence. For instance, the term $\zket + 0 \cdot (\lbd y y) \oket$ does not 
reduce and is not a value.

Finally, typing is also preserved by reduction, provided we consider pure or
terminating terms.

\begin{restatable}[Subject reduction]{theorem}{subred}
  \label{lem:subred}
  Let $\typing \Gamma \Delta t T$ be a well-typed term, and $t \reduces t'$.
  If $t$ terminates or $t$ is pure,
  then $\typing \Gamma \Delta {t'} T$.
\end{restatable}

\begin{remark}\label{rem:subred}
  Theorem~\ref{lem:subred} does not hold on non-terminating and non-pure
  terms, as we cannot type non-terminating superpositions.
  This is the case with the following example, with the Bell state
  $\ket{\Phi^+} =
  \frac{1}{\sqrt 2} \cdot \zket\otimes \zket
  + \frac{1}{\sqrt 2} \cdot \oket\otimes \oket$:
  \[
    \left(\lbd y \qcase{
        ((\letrec f x fx)\zket)
    }{\zket \otimes y}{\oket \otimes y}\right) \ket{\Phi^+}
  \]
  While such a term is well-typed, it reduces to a superposition of
  non-terminating pure terms, which is impossible to type through the rule
  \hyperlink{typesup}{(sup)} from Table~\ref{tab:typing}, as
  the orthogonality predicate requires termination. This is a counterpart of
  having a general language with few restrictions; however, as resource
  analysis will be done on terminating terms, such ill-behaved cases will not be
  considered.
\end{remark}

\subsection{Decidability of Orthogonality}
\label{subsec:orthogonality}
The orthogonality predicate of Definition~\ref{def:orthogonality} is
extensionally complete. As a consequence, it is undecidable.
Indeed, it is as hard as the Universal Halting Problem~\cite{EGZ09}, that is
$\Pi^0_2$-complete in the arithmetical hierarchy~\cite{Nie09}.
Orthogonality uses the inner product, which requires computing
additions, multiplications, and nullity checks on complex numbers. Recall that
algebraic numbers are complex numbers in $\mathbb C$ that are roots of a
polynomial in $\mathbb Q[X]$, and write their sets as $\bar{\mathbb C}$. In the
field of algebraic numbers, equality is decidable, and product and sum are
computable in polynomial time~\cite{HHK05}. We thus restrict ourselves to terms
that only contain amplitudes $\alpha \in \bar{\mathbb C}$.

\begin{restatable}[Undecidability of orthogonality]{theorem}{orthoundecide}
  \label{thm:orthogonality-undecidable}
  Deciding orthogonality between two well-typed terms is $\pitwo$-complete.
\end{restatable}

This is not surprising as \hyrql{} allows for
general recursion. Decidability can be recovered under some slight restrictions.
Towards that end, we introduce for basic types
a \emph{type depth} $d$ as a partial function $d: \{\qbit \}\cup\bset \to
\mathbb N $, defined as follows:
\[
  d(\qbit) \triangleq 1 \qquad
  d(B) \triangleq \max_{\substack{{ s = \sign{\bkvar_1, \dots, \bkvar_n}
  \bvar} \\ c^s \in \cons B}}
  (\sum_{i=1}^n d(\bkvar_i)) + 1
\]
Note that $d(\bkvar)$ is undefined on inductive types.
We call \emph{finite types} any basic type with a defined depth.

\begin{restatable}{proposition}{ortholower}
  \label{prop:orthogonality-lower}
  Given a finite type $\bkvar$, and two terminating terms $\closedtyping s
  \bkvar$, $\closedtyping t \bkvar$, it is decidable whether they are
  orthogonal.
\end{restatable}

In previous works of \spm~\cite{CSV23, CLV24},
the syntax of the language was restrictive enough to ensure the decidability of
orthogonality. As \spm{} is a strict sublanguage of {\hyrql}, orthogonality
can also be decided on this fragment.

\subsection{Complexity Properties}\label{subsec:poly}

We finally discuss how quantum circuits and their size are related to terms of
\hyrql{} and their runtime-complexity. Towards that end, we introduce a
sublanguage of \hyrql{} ($\hyrqlbound T$, Definition~\ref{def:hyrqlbound})
based on three restrictions:
one on the runtime-complexity of the program, to ensure terminating terms
($\timeset T$, Definition~\ref{def:time});
one on types, to obtain terms that can be represented as circuits
($\circuittype$, Definition~\ref{def:restrict-circuit});
and one on recursive calls, to have a faithful relation between the size of the
circuit and the runtime-complexity ($\boundedrec$, Definition~\ref{def:brec}).
This allows, for terms terminating in polynomial time, to characterize
precisely the class \fbqp{}. We first define the time termination of a term.

\begin{definition}\label{def:time}
  Given a term $\closedtyping t Q \linfunc Q'$ and a function $T : \mbN \to \mbN$,
  we say that $t$ \emph{terminates in time $T$}, if for any value 
  $\closedtyping v Q$, $tv$ terminates in at most $T(\size v)$ steps.
  Let $\timeset T$ be the set of terms terminating in time $T$.
\end{definition}

Note that this definition should not be confused with termination defined in
Section~\ref{subsec:semantics}.

To obtain terms that can be compiled to a circuit form, we introduce a 
typed restriction \circuittype{}, which relies on the following type grammar:
\[
    B_r \Coloneqq \qbit \gmid \nat \gmid \typelist{B_r} \gmid B_r \otimes
    B_r \gmid B_r \linfunc B_r \gmid B_r \nonlinfunc B_r
\]

\begin{definition}[Circuit terms]\label{def:restrict-circuit}
  We define $\circuittype$ as the set of well-typed terms $\pi \triangleright
  \closedtyping s Q \linfunc Q'$ satisfying the following conditions:
  \begin{itemize}
    \item Any type in $\pi$ is generated by the grammar of $B_r$;
    \item Any rule $\pi \rightslice (\mathrm{sup})[t_1, \dots, t_n]$
      satisfies $t_i \in \valueset$;
    \item Any rule $\pi \rightslice (\mathrm{qcase})[t, t_0, t_1]$ satisfies
      either $t_0, t_1 \in \valueset$; or $t_0 = c(\ket 0, s_0)$ and
      $t_1 = c(\ket 1, s_1)$, for $c \in \set{\otimes, ::}$.
  \end{itemize}
\end{definition}

While the first condition yields terms that can be represented
by a circuit, the two others allow for an efficient compilation and 
guarantees that any superposition can be compiled effectively to circuits.
Note that while these terms only take inputs of type $Q$, to represent
them clearly in a quantum circuit, they can nonetheless contain
multiple inputs and higher-order constructs in their body.
This restriction alone is sufficient to compile $t$ to a family of circuits in
Theorem~\ref{thm:circuit-exist}, provided $t$ terminates for a given time $T$.

\begin{example}
  \label{ex:bqwalk-adapt}
  The following term, which is an adaptation of the program from
  Example~\ref{ex:bqwalk}, belongs to $\circuittype$. 
  \[
    \pi \triangleright \closedtyping {\lbd x \match x {q \otimes n \to
    \bqwalk\,q\,n}} {\qbit \otimes \nat \linfunc \typelist \qbit}
  \]
  Any type in $\pi$ is derived from $B_r$. Term superposition happens only in
  $\mathtt{Had}$ and is between values. Furthermore, there are two qcase
  constructs, either between values, or between two terms $\zket :: s_0$ and
  $\oket :: s_1$, thus it satisfies all properties of $\circuittype$. Note that
  as $B_r$ allows for tensor products, any program taking multiple inputs can
  always be rewritten using the same technique so that it belongs to
  $\circuittype$.
\end{example}

In the literature, family of circuits are parameterized by the size
of their input. However, in \hyrql{}, knowing the size might not be sufficient,
e.g., on inputs of type $\typelist \qbit \otimes \typelist \qbit$. 
To remedy this, we will parameterize the family by the \emph{shape} of the
input, which will give us the necessary information. Formally, we define
\[
  \begin{aligned}
    \shapeset Q &\triangleq \set{w \mid \closedtyping w \shape Q} \\
    \shapeset w &\triangleq \set{v \mid \closedtyping v Q \wedge \shape v
    \sreduces w}, \forall w \in \shapeset Q,
  \end{aligned}
\]
where $\shapeset Q$ is the \emph{set of shape of terms of type $Q$}, and
$\shapeset w$ the \emph{set of values of shape $w$}.

Assuming an encoding of $t \in \circuittype$ to a quantum state $\ket t$, which
can be done naturally, a circuit $C$ is said to \emph{approximate} $v$ on input
$v'$ with probability $p$, if $C$ evaluates to $\ket \phi$, on input $\ket{v'}$
(up to ancillary qubits), and $\lvert \langle v | \phi \rangle \rvert^2 \geq
p$.

\begin{restatable}[Circuit compilation]{theorem}{circuitexistence}
  \label{thm:circuit-exist}
  Let $t \in \circuittype \cap \timeset T$ be a term of type $Q \linfunc Q'$
  for a given $T : \mbN \to \mbN$. 
  Then, we can generate a family of circuits $\famcirc \triangleq (C_w)_{w \in
  \shapeset Q}$ such that for any $w \in \shapeset Q$ and any $v \in \shapeset
  w$, $C_w$ approximates $v'$ with probability $\frac 2 3$ on input $v$, where
  $tv \sreduces v'$.
\end{restatable}

We are now interested in the \emph{size} of each circuit of the family, i.e.,
the total number of gates acting on the circuit; in particular, the goal is to
exhibit a relation between the size and the runtime of the initial term $t$. To
avoid a blowup due to recursive calls, we introduce a syntactic restriction
$\boundedrec$. It relies on the notion of the \emph{width} of a variable $f$ in
a term $t$, written $w(f, t)$ and defined inductively as follows:
\[
\begin{array}{c}
  w(f, x) \triangleq \delta_{x, f} 
  \qquad 
  w(f, \ket i) \triangleq 0 
  \qquad 
  w(f, \lbd x t) \triangleq w(f,t) \qquad
  w(f, \letrec g x t) \triangleq w(f, t) \qquad \\[1ex]
  w(f, t_1 t_2) \triangleq w(f, t_1) + w(f, t_2)
  \qquad 
  w(f, \sum_{i=1}^n \alpha_i \cdot t_i) = \max_{1 \leq i \leq n} w(f, t_i) \\[1ex]
  w(f, c(t_1, \dots, t_n))       \triangleq \sum_{i=1}^n w(f, t_i) 
  \qquad
  w(f, \shape t) = w(f, t) \\[1ex]
  w(f,\qcasedefault t) \triangleq w(f, t) + \max(w(f,t_0), w(f,t_1)) \\[1ex]
  w(f, \matchdefaultcollapsed t) \triangleq w(f, t) + \max_{1 \leq i \leq
  n} (w(f, t_i))
\end{array}
\]

\begin{definition}\label{def:brec}
  A term $t$ is said to \emph{compile faithfully}, if for any term $t'=(\letrec
  f x s)$ such that $t' \subtermeq t$, the following conditions hold:
  \begin{itemize}
    \item $w(f, s) \leq 1$;
    \item $\exists !\  t', \forall\ t_1, f t_1 \subtermeq s \implies t_1 = t'$;
    \item $\neg \exists\  t', t'f \subtermeq s$.
  \end{itemize}
  Let $\boundedrec$ be the set of terms that compile faithfully.
\end{definition}

The first two conditions are derived from \cite{HPS25}, and ensure that, for
each recursive program $\letrec f x s$, all recursive calls are done in
different branches of a \textqcase{} or \textmatch{}, with the same input.
Furthermore, the last condition ensures that this behaviour is preserved, even
in presence of higher-order constructs.

\begin{example}
  The program of Example~\ref{ex:bqwalk-adapt} belongs to
  $\boundedrec$, as both $\bqwalk$ and $\tt repeat$ only perform one recursive
  call. All the examples across the paper also satisfy this criterion.
\end{example}

\begin{definition}\label{def:hyrqlbound}
  Let $T : \mbN \to \mbN$. We define $\hyrqlbound T$ as follows:
  \[
    \hyrqlbound T\ \triangleq\ \circuittype\ \cap\ \boundedrec\ \cap\ \timeset T
  \]
\end{definition}

Any term in this subset compiles to quantum circuits of size bounded
polynomially by $T$.

\begin{restatable}[Bounded circuit compilation]{theorem}{circuitbound}
  \label{thm:circuit-bound}
  Let $T : \mbN \to \mbN$, and let $t \in \hyrqlbound T$ be a term of type
  $Q \linfunc Q'$. Then, we can generate a family of circuits
  $\famcirc \triangleq (C_w)_{w \in \shapeset Q}$ and a polynomial $P$ such
  that, for any $w \in \shapeset Q$:
  \begin{itemize}
    \item $\size{C_w} \leq P(T(\size w))$;
    \item $\forall\ v \in \shapeset w$, $C_w$ approximates $v'$ with
      probability $\frac 2 3$ on input $v$, where $tv \sreduces v'$.
  \end{itemize}
\end{restatable}

We now want to use Theorem~\ref{thm:circuit-bound} to obtain a characterization
of \fbqp{} (Definition~\ref{def:fbqp}). We recall that a family of circuits
$(C_n)_{n \in \mbN}$ is said to be \emph{uniform polynomially-sized} with
$\alpha : \mbN \to \mbN$ if there exists $P \in \mbN[X]$ such that $\size{C_n}
\leq P(n)$ and $C_n$ has exactly $\alpha(n)$ ancillary qubits, and there is a
polynomial-time Turing machine, which takes $n$ as input, and outputs a
representation of $C_n$ for any $n \in \mbN$.

\begin{definition}[\cite{BV97}]
  \label{def:fbqp}
  A binary function $f : \binfunc$ is said to be computed by a family of
  circuits $(C_n)_{n \in \mbN}$ if for any $x \in \set{0,1}^*$, $C_{\size x}$
  approximates $f(x)$ with probability $\frac 2 3$ on input $x$. \fbqp{} is
  defined as the set of binary functions computed by a uniform
  polynomially-sized family of circuits.
\end{definition}

To characterize \fbqp{}, we define the following set of terms, which terminate 
in polynomial time, and have list of qubits as input:
\[
  \hyrqlpoly \triangleq \cup_{P \in \mbN[X]} \set{
    t \in \hyrqlbound P \gmid \closedtyping t \typelist \qbit \linfunc Q
  }
\]

These generated families of circuits will be indexed by natural numbers
instead of shape, as one can remark that $\shapeset{\typelist \qbit}$ is in
bijection with $\mbN$.
Altogether, $\hyrqlpoly$ contain terms that are
compiled to circuits approximating functions of $\fbqp$.

\begin{restatable}[\fbqp{} soundness]{theorem}{fbqpsound}
  \label{thm:fbqp-sound}
  Let $t \in \hyrqlpoly$, and let $f : \binfunc$. If $f$ is computed by
  $\famcirc$, then $f \in \fbqp$.
\end{restatable}

For example, the Quantum Fourier Transform program from Example~\ref{ex:qft}
belongs to $\fbqp$, and can thus be compiled to circuits of polynomial size.

We now prove \fbqp{} completeness, meaning that every function of \fbqp{} can
be represented by a function in $\hyrqlpoly$. This is proven by writing
any function of Yamakami's algebra~\cite{Yam20}, which is known to capture
\fbqp{}, into $\hyrqlpoly$.

\begin{restatable}[\fbqp{} completeness]{theorem}{fbqpcomplete}
  \label{thm:fbqp-complete}
  Let $f : \binfunc \in \fbqp$. Then, there exists $t \in \hyrqlpoly$
  such that $\famcirc$ computes $f$.
\end{restatable}

We have given a typed, syntactic and semantic restriction of \hyrql{}, that
compiles into a bounded family of quantum circuits. Such restriction could be
refined, either to accept more terms by lifting requirements in
Definition~\ref{def:brec}, as done in~\cite{HPS25}; or by adding requirements
to characterize more precise complexity classes~\cite{FHPS25}. We are
nonetheless able to encapsulate \fbqp{}, which demonstrates the expressivity of
the $\hyrqlpoly$ fragment.

\section{Towards Automatization}
\label{sec:automatization}
Theorem~\ref{thm:circuit-bound} establishes that, given a program
that terminates with a specified complexity bound, one can extract a
concrete quantum circuit exhibiting the same space complexity. This result
is particularly significant, as it enables the characterization of
 quantum polynomial time, as shown in Theorem~\ref{thm:fbqp-sound} and
 Theorem~\ref{thm:fbqp-complete}. However, the burden of proving termination of
 a \hyrql{} program currently rests entirely on the programmer.
This naturally raises the following question: to what
extent is the termination of \hyrql{} programs semi automatable?

We argue that such automatization is feasible by leveraging techniques
from Term Rewrite Systems (TRS)\footnote{More specifically, a higher-order version
called Simply-Typed TRS~\cite{Yam01}, but we do not enter this level of details
here.}, as they
are equipped with a rich toolbox for analyzing both termination and
computational complexity. Approaches to termination analysis include
non-exhaustively recursive path
orderings~\cite{Der82,Der87}, dependency pairs~\cite{AG00}, and
size-change termination~\cite{LJBA01}. Similarly, several methods
address complexity analysis, e.g.,
interpretation methods~\cite{BMM11,MP09,BDL16}, and polynomial path
orders~\cite{Moser09,AM09}.

Although a full formal treatment is left to future work, we
present an illustrative example supporting our claim. Specifically, we
illustrate how a \hyrql{} program can be efficiently translated into a (quantum) TRS,
how the reduction behaviour of the original program is faithfully
simulated by the resulting rewriting system while preserving
complexity bounds, and how termination and complexity results
established for the TRS can be transferred back to the corresponding
\hyrql{} program.

For example, the standard NOT gate is represented by the quantum TRS whose set
of rewrite rules is $ \{\mathtt{NOT}\,\zket \to \oket, \mathtt{NOT}\,\oket \to
\zket\}$. The application of the $\mathtt{NOT}$ function symbol to $\zket$
rewrites to $\oket$. Similarly, the Hadamard gate $H$ can be defined using a
superposition of terms $ \{\mathtt{H}\,\zket \to \frac{1}{\sqrt 2} \cdot \zket
+ \frac{1}{\sqrt 2} \cdot \oket, \mathtt{H}\,\oket \to  \frac{1}{\sqrt 2} \cdot
\zket - \frac{1}{\sqrt 2} \cdot \oket\}$. The evaluation of a quantum term will
follow a parallel strategy, similar to \hyrql{} operational semantics
(Table~\ref{tab:reduction}), which reduces in parallel the distinct branches of
a superposition, when possible (otherwise, the reduction will have an
exponential blowup).

We claim that we can automatize the translation of any term of \hyrql{} to a
corresponding TRS. This translation works as follows:

\begin{itemize}
  \item Produce a set of rules, inductively on the syntax of the term;
  \item Yield a function symbol $\mtf$ for the initial term $t$;
  \item When encountering $\qcasedefault x$ (same for \textmatch{}), merge the 
    two sets of rules generated for $t_0$ and $t_1$, while keeping track of
    which path was taken.
\end{itemize}

Coming back on Example~\ref{ex:bqwalk}, such an algorithm would generate the
following rewrite rules:
\[
  \left\{
    \begin{aligned}
      \mathtt{Repeat}\,0 &\to \nil   &\mathtt{BQWalk}\,\zket\,m &\to \zket ::
      \mathtt{Repeat}\,m \\
      \mathtt{Repeat}\,S(m) &\to \zket:: \mathtt{Repeat}\,m  &
      \mathtt{BQWalk}\,\oket\,0 &\to \oket :: \nil \\
      & &     \mathtt{BQWalk}\,\oket\,S(m) &\to \oket ::
      \mathtt{BQWalk}\,(\ket -)\,m
    \end{aligned}
  \right\}
\]

One can remark that this set of rules captures completely
the behaviour of $\mathtt{bqwalk}$, and they both compute the same function;
furthermore, one reduction step for a given rewrite rule will correspond to
$1 \leq k \leq \size{\mathtt{bqwalk}}$ reduction steps in \hyrql{}.
Therefore, a certificate on this TRS yields directly a certificate on the
original term. 

Now, it is possible to apply the tools we mentioned to the obtained TRS. In
particular, we can use recursive path orderings~\cite{Der82} to show
termination, coupled with quasi-interpretations~\cite{BMM11} to get that this
TRS terminates in time $\mcal O(n)$ for any input of size $n$. While such
properties are not defined for quantum TRS, i.e.,\ with superpositions and
parallel reduction, they can be adapted to fit the quantum setting; a formal
adaptation is left for future work. Altogether with
Theorem~\ref{thm:circuit-bound}, we can conclude that
$\mathtt{bqwalk}$ can be compiled to a family of quantum circuits, whose size
is polynomial in the size of the input.

Using the same ideas, we claim that we can describe a semantic-preserving
translation algorithm from \hyrql{} to TRS, with no blowup on the termination
complexity in either direction. Such work, along with a proper development of
TRS in the quantum setting and an adaptation of existing techniques, would give
a semi-automatized way to prove the feasibility/tractability of terms in
\hyrql{}.

While this section was focused on termination and time complexity, some of the
mentioned techniques also provide results to bound the space complexity of a
TRS program~\cite{BMM11}. Given a good definition of quantum space, this would
make it possible to characterize quantum space complexities, and obtain bounds
on the depth or width of the circuit.

\bibliography{bibli}

@article{ACB14,
  title         = {Computational Advantage from Quantum-Controlled Ordering of Gates},
  author        = {Ara\'{u}jo, Mateus and Costa, Fabio and Brukner, \v{C}aslav},
  journal       = {Phys. Rev. Lett.},
  volume        = {113},
  issue         = {25},
  pages         = {250402},
  numpages      = {5},
  year          = {2014},
  month         = {dec},
  publisher     = {American Physical Society},
  doi           = {10.1103/PhysRevLett.113.250402}
}

@article{AG00,
  title         = {Termination of term rewriting using dependency pairs},
  journal       = {Theoretical Computer Science},
  volume        = {236},
  number        = {1},
  pages         = {133--178},
  year          = {2000},
  issn          = {0304-3975},
  doi           = {10.1016/S0304-3975(99)00207-8},
  author        = {Thomas Arts and J\"{u}rgen Giesl}
}

@inproceedings{AG05,
  author        = {Altenkirch, Thorsten and Grattage, Jonathan},
  title         = {A Functional Quantum Programming Language},
  booktitle     = {Logic in Computer Science, LICS'05},
  year          = {2005},
  pages         = {249--258},
  publisher     = {IEEE Computer Society},
  doi           = {10.1109/lics.2005.1}
}

@inproceedings{AM09,
  title         = {Dependency pairs and polynomial path orders},
  author        = {Avanzini, Martin and Moser, Georg},
  booktitle     = {International Conference on Rewriting Techniques and Applications, RTA'09},
  pages         = {48--62},
  year          = {2009},
  doi           = {10.1007/978-3-642-02348-4_4},
  organization  = {Springer}
}

@inproceedings{AMPPZ22,
  author        = {Martin Avanzini and Georg Moser and Romain P{\'{e}}choux and Simon Perdrix and Vladimir Zamdzhiev},
  editor        = {Christel Baier and Dana Fisman},
  title         = {Quantum Expectation Transformers for Cost Analysis},
  booktitle     = {{LICS} '22: Symposium on Logic in Computer Science},
  pages         = {10:1--10:13},
  publisher     = {{ACM}},
  year          = {2022},
  doi           = {10.1145/3531130.3533332}
}

@article{AWHM20,
  title         = {Communication through coherent control of quantum channels},
  volume        = {4},
  issn          = {2521-327X},
  doi           = {10.22331/q-2020-09-24-333},
  journal       = {Quantum},
  publisher     = {Verein zur Forderung des Open Access Publizierens in den Quantenwissenschaften},
  author        = {Abbott, Alastair A. and Wechs, Julian and Horsman, Dominic and Mhalla, Mehdi and Branciard, Cyril},
  year          = {2020},
  month         = sep,
  pages         = {333}
}

@article{Bar84,
  title         = {The lambda calculus: its syntax and semantics},
  author        = {Barendregt, Henk},
  journal       = {Studies in logic and the foundations of Mathematics},
  year          = {1984},
  doi           = {10.1016/c2009-0-14341-6}
}

@article{BB84,
  title         = {Quantum cryptography: Public key distribution and coin tossing},
  journal       = {Theoretical Computer Science},
  volume        = {560},
  pages         = {7--11},
  year          = {2014},
  issn          = {0304-3975},
  doi           = {10.1016/j.tcs.2014.05.025},
  author        = {Charles H. Bennett and Gilles Brassard}
}

@article{BDL16,
  title         = {Higher-Order Interpretations and Program Complexity},
  author        = {Baillot, Patrick and Dal Lago, Ugo},
  year          = {2016},
  journal       = {Information and Computation},
  volume        = {248},
  pages         = {56-81},
  publisher     = {Elsevier},
  doi           = {10.1016/j.ic.2015.12.008}
}

@article{BMM11,
  title         = {Quasi-Interpretations a Way to Control Resources},
  author        = {Bonfante, Guillaume and Marion, Jean-Yves and Moyen, Jean-Yves},
  year          = {2011},
  month         = jun,
  journal       = {Theoretical Computer Science},
  volume        = {412},
  number        = {25},
  pages         = {2776--2796},
  issn          = {0304-3975},
  doi           = {10.1016/j.tcs.2011.02.007}
}

@misc{BPP25,
  title         = {A Quantum Programming Language for Coherent Control},
  author        = {Kathleen Barsse and Romain P\'{e}choux and Simon Perdrix},
  year          = {2025},
  eprint        = {2507.10466},
  archiveprefix = {arXiv},
}

@article{BV97,
  title         = {Quantum {{Complexity Theory}}},
  author        = {Bernstein, Ethan and Vazirani, Umesh},
  year          = 1997,
  month         = oct,
  journal       = {SIAM Journal on Computing},
  publisher     = {{Society for Industrial and Applied Mathematics}},
  volume        = 26,
  number        = 5,
  pages         = {1411--1473},
  doi           = {10.1137/S0097539796300921},
  issn          = {0097-5397}
}

@article{CAPV13,
  title         = {Quantum computations without definite causal structure},
  author        = {Chiribella, Giulio and D'Ariano, Giacomo Mauro and Perinotti, Paolo and Valiron, Benoit},
  journal       = {Phys. Rev. A},
  volume        = {88},
  issue         = {2},
  pages         = {022318},
  numpages      = {15},
  year          = {2013},
  month         = {Aug},
  publisher     = {American Physical Society},
  doi           = {10.1103/PhysRevA.88.022318}
}

@article{CDL25,
  author        = {Colledan, Andrea and Dal Lago, Ugo},
  title         = {Flexible Type-Based Resource Estimation in Quantum Circuit Description Languages},
  year          = {2025},
  issue_date    = {January 2025},
  publisher     = {Association for Computing Machinery},
  address       = {New York, NY, USA},
  volume        = {9},
  number        = {POPL},
  doi           = {10.1145/3704883},
  journal       = {Proc. ACM Program. Lang.},
  month         = jan,
  articleno     = {47},
  numpages      = {31}
}

@inproceedings{CLV24,
  title         = {Semantics for a {{Turing-Complete Reversible Programming Language}} with {{Inductive Types}}},
  booktitle     = {{{Formal Structures}} for {{Computation}} and {{Deduction}} ({{FSCD}} 2024)},
  author        = {Chardonnet, Kostia and Lemonnier, Louis and Valiron, Beno{\^i}t},
  editor        = {Rehof, Jakob},
  year          = {2024},
  series        = {{LIPIcs}},
  volume        = {299},
  pages         = {19:1--19:19},
  issn          = {1868-8969},
  doi           = {10.4230/LIPIcs.FSCD.2024.19},
  isbn          = {978-3-95977-323-2}
}

@article{CSV23,
  title         = {A {{Curry-Howard Correspondence}} for {{Linear}}, {{Reversible Computation}}},
  author        = {Chardonnet, Kostia and Saurin, Alexis and Valiron, Beno{\^i}t},
  year          = {2023},
  journal       = {LIPIcs, Volume 252, CSL 2023},
  volume        = {252},
  pages         = {13:1--13:18},
  issn          = {1868-8969},
  doi           = {10.4230/LIPICS.CSL.2023.13},
  isbn          = {9783959772648}
}

@article{DCM22,
  author        = {Alejandro D{\'{\i}}az{-}Caro and Octavio Malherbe},
  title         = {Quantum Control in the Unitary Sphere: Lambda-S1 and its Categorical Model},
  journal       = {Log. Methods Comput. Sci.},
  volume        = {18},
  number        = {3},
  year          = {2022},
  doi           = {10.46298/LMCS-18(3:32)2022}
}

@article{Der82,
  title         = {Orderings for term-rewriting systems},
  author        = {Dershowitz, Nachum},
  journal       = {Theoretical computer science},
  volume        = {17},
  number        = {3},
  pages         = {279--301},
  year          = {1982},
  doi           = {10.1016/0304-3975(82)90026-3},
  publisher     = {Elsevier}
}

@article{Der87,
  author        = {Dershowitz, Nachum},
  title         = {Termination of rewriting},
  year          = {1987},
  issue_date    = {Feb./April 1987},
  publisher     = {Academic Press, Inc.},
  address       = {USA},
  volume        = {3},
  number        = {1-2},
  issn          = {0747-7171},
  doi           = {10.1016/S0747-7171(87)80022-6},
  journal       = {J. Symb. Comput.},
  month         = feb,
  pages         = {69-115},
  numpages      = {47}
}

@article{DLMZ10,
  author        = {Dal Lago, Ugo and Masini, Andrea and Zorzi, Margherita},
  title         = {Quantum implicit computational complexity},
  journal       = {Theor. Comput. Sci.},
  volume        = {411},
  number        = {2},
  pages         = {377--409},
  year          = {2010},
  doi           = {10.1016/J.TCS.2009.07.045}
}

@inproceedings{DLPZ25,
  author        = {Dave, Kinnari and Lemonnier, Louis and P\'{e}choux, Romain and Zamdzhiev, Vladimir},
  title         = {Combining quantum and classical control: syntax, semantics and adequacy},
  year          = {2025},
  isbn          = {978-3-031-90896-5},
  publisher     = {Springer-Verlag},
  address       = {Berlin, Heidelberg},
  doi           = {10.1007/978-3-031-90897-2_8},
  booktitle     = {Foundations of Software Science and Computation Structures, FoSSaCS 2025},
  pages         = {155-175},
  numpages      = {21},
  location      = {Hamilton, ON, Canada}
}

@inbook{EGZ09,
  title         = {Degrees of Undecidability in Term Rewriting},
  isbn          = {9783642040276},
  issn          = {1611-3349},
  doi           = {10.1007/978-3-642-04027-6_20},
  booktitle     = {Computer Science Logic},
  publisher     = {Springer Berlin Heidelberg},
  author        = {Endrullis,  J\"{o}rg and Geuvers,  Herman and Zantema,  Hans},
  year          = {2009},
  pages         = {255-270}
}

@inproceedings{FHPS25,
  title         = {Quantum Programming in Polylogarithmic Time},
  author        = {Florent Ferrari and Emmanuel Hainry and Romain P\'{e}choux and M\'{a}rio Silva},
  booktitle     = {Mathematical Foundations of Computer Science, {MFCS} 2025},
  series        = {LIPIcs},
  volume        = {345},
  pages         = {47:1--47:17},
  doi           = {10.4230/LIPICS.MFCS.2025.47},
  year          = {2025}
}

@article{FY21,
  author        = {Feng, Yuan and Ying, Mingsheng},
  journal       = {ACM Transactions on Quantum Computing},
  number        = {4},
  pages         = {1--43},
  publisher     = {ACM New York, NY},
  title         = {Quantum {Hoare} logic with classical variables},
  volume        = {2},
  doi           = {10.1145/3456877},
  year          = {2021}
}

@inproceedings{GLRSV13,
  title         = {Quipper: a scalable quantum programming language},
  author        = {Green, Alexander S. and Lumsdaine, Peter LeFanu and Ross, Neil J. and Selinger, Peter and Valiron, Beno{\^i}t},
  doi           = {10.1145/2491956.2462177},
  booktitle     = {{ACM} {SIGPLAN} Conference on Programming Language Design and Implementation, {PLDI} '13},
  publisher     = {ACM},
  year          = {2013},
  month         = jun,
  pages         = {333-342}
}

@techreport{HHK05,
  author        = {Halava, Vesa and Harju, Tero and Hirvensalo, Mika and Karhum{\"a}ki, Juhani},
  title         = {Skolem's Problem - On the Border between Decidability and Undecidability},
  institution   = {Turku Center for Computer Science},
  year          = 2005,
  number        = {683},
  comment       = {According to \cite{OuaknineWor12}, "Unfortunately, it appears that the proof is incorrect, with no immediately apparent way to repair it. The critical case which is not adequately handled by the authors is that of a linear recurrence sequence whose characteristic polynomial has five distinct roots, four of which are complex and of the same magnitude, and one of which is real and of strictly smaller magnitude."},
  url           = {https://tucs.fi/publications/view/?pub_id=tHaHaHiKa05a}
}

@inproceedings{HPS23,
  title         = {A Programming Language Characterizing Quantum Polynomial Time},
  author        = {Hainry, Emmanuel and P{\'e}choux, Romain and Silva, M{\'a}rio},
  year          = {2023},
  month         = apr,
  pages         = {156--175},
  doi           = {10.1007/978-3-031-30829-1_8},
  booktitle     = {Foundations of Software Science and Computation Structures, Fossacs 2023},
  publisher     = {Springer},
  isbn          = {978-3-031-30828-4}
}

@inproceedings{HPS25,
  author        = {Emmanuel Hainry and Romain P{\'{e}}choux and M{\'{a}}rio Silva},
  editor        = {Maribel Fern{\'{a}}ndez},
  title         = {Branch Sequentialization in Quantum Polytime},
  booktitle     = {Formal Structures for Computation and Deduction, {FSCD} 2025},
  series        = {LIPIcs},
  volume        = {337},
  pages         = {22:1--22:22},
  year          = {2025},
  doi           = {10.4230/LIPICS.FSCD.2025.22}
}

@article{Kit97,
  title = {Quantum computations: algorithms and error correction},
  volume = {52},
  ISSN = {1468-4829},
  DOI = {10.1070/rm1997v052n06abeh002155},
  number = {6},
  journal = {Russian Mathematical Surveys},
  publisher = {Steklov Mathematical Institute},
  author = {Kitaev, Alexei Y.},
  year = {1997},
  month = dec,
  pages = {1191–1249}
}

@techreport{Kni22,
  title         = {Conventions for quantum pseudocode},
  author        = {Knill, Emanuel},
  eprint        = {2211.02559},
  archiveprefix = {arXiv},
  institution   = {Los Alamos National Lab.},
  year          = {1996}
}

@misc{KOY24,
  title         = {Exponential separation in quantum query complexity of the quantum switch with respect to simulations with standard quantum circuits},
  author        = {Hl\'{e}r Kristj\'{a}nsson and Tatsuki Odake and Satoshi Yoshida and Philip Taranto and Jessica Bavaresco and Marco T\'{u}lio Quintino and Mio Murao},
  year          = {2024},
  eprint        = {2409.18420},
  archiveprefix = {arXiv},
}

@article{L04,
  author        = {Yves Lafont},
  title         = {Soft linear logic and polynomial time},
  journal       = {Theor. Comput. Sci.},
  volume        = {318},
  number        = {1-2},
  pages         = {163--180},
  year          = {2004},
  doi           = {10.1016/J.TCS.2003.10.018}
}

@phdthesis{Lem24,
  title         = {The {{Semantics}} of {{Effects}} : {{Centrality}}, {{Quantum Control}} and {{Reversible Recursion}}},
  shorttitle    = {The {{Semantics}} of {{Effects}}},
  author        = {Lemonnier, Louis},
  year          = {2024},
  month         = jun,
  school        = {Universit{\'e} Paris-Saclay}
}

@inproceedings{LJBA01,
  author        = {Lee, Chin Soon and Jones, Neil D. and Ben-Amram, Amir M.},
  title         = {The size-change principle for program termination},
  booktitle     = {Symposium on Principles of Programming Languages, POPL 2001},
  year          = {2001},
  publisher     = {ACM},
  address       = {New York, NY, USA},
  doi           = {10.1145/373243.360210},
  month         = jan,
  pages         = {81-92},
  numpages      = {12}
}

@phdthesis{Moser09,
  author        = {Georg Moser},
  title         = {Proof Theory at Work: Complexity Analysis of Term Rewrite Systems},
  year          = {2009},
  eprint        = {0907.5527},
  school        = {University of Innsbruck},
  type          = {Cumulative Habilitation Thesis}
}

@article{MP09,
  author        = {Marion, Jean-Yves and P\'{e}choux, Romain},
  title         = {Sup-interpretations, a semantic method for static analysis of program resources},
  year          = {2009},
  issue_date    = {August 2009},
  publisher     = {Association for Computing Machinery},
  address       = {New York, NY, USA},
  volume        = {10},
  number        = {4},
  issn          = {1529-3785},
  doi           = {10.1145/1555746.1555751},
  journal       = {ACM Trans. Comput. Logic},
  month         = aug,
  articleno     = {27},
  numpages      = {31}
}

@book{NC12,
  title        = {Quantum Computation and Quantum Information: 10th Anniversary Edition},
  author       = {Nielsen,  Michael A. and Chuang,  Isaac L.},
  year         = 2012,
  month        = jun,
  publisher    = {Cambridge University Press},
  doi          = {10.1017/cbo9780511976667},
  isbn         = 9780511976667,
}

@book{Nie09,
  author        = {Nies, Andr\'{e}},
  title         = {Computability and Randomness},
  publisher     = {Oxford University Press},
  year          = {2009},
  month         = {01},
  isbn          = {9780199230761},
  doi           = {10.1093/acprof:oso/9780199230761.001.0001}
}

@phdthesis{Pec20,
  title         = {{Complexit{\'e} implicite : bilan et perspectives}},
  author        = {P{\'e}choux, Romain},
  url           = {https://hal.univ-lorraine.fr/tel-02978986},
  school        = {{Universit{\'e} de Lorraine}},
  year          = {2020},
  month         = Oct,
  type          = {Accreditation to supervise research},
  hal_id        = {tel-02978986},
  hal_version   = {v3}
}

@article{PMA+15,
  title         = {Experimental superposition of orders of quantum gates},
  volume        = {6},
  issn          = {2041-1723},
  doi           = {10.1038/ncomms8913},
  number        = {1},
  journal       = {Nature Communications},
  publisher     = {Springer Science and Business Media LLC},
  author        = {Procopio, Lorenzo M. and Moqanaki, Amir and Ara\'{u}jo, Mateus and Costa, Fabio and Alonso Calafell, Irati and Dowd, Emma G. and Hamel, Deny R. and Rozema, Lee A. and Brukner, \v{C}aslav and Walther, Philip},
  year          = {2015},
  month         = aug
}

@misc{QASM17,
  archiveprefix = {arXiv},
  author        = {Andrew W. Cross and Lev S. Bishop and John A. Smolin and Jay M. Gambetta},
  eprint        = {1707.03429},
  primaryclass  = {quant-ph},
  title         = {Open Quantum Assembly Language},
  year          = {2017},
}

@article{QASM22,
  author        = {Cross, Andrew W. and Javadi-Abhari, Ali and Alexander, Thomas and De Beaudrap, Niel and Bishop, Lev S. and Heidel, Steven and Ryan, Colm A. and Sivarajah, Prasahnt and Smolin, John and Gambetta, Jay M. and Johnson, Blake R.},
  doi           = {10.1145/3505636},
  issn          = {2643-6817},
  journal       = {ACM Transactions on Quantum Computing},
  month         = sep,
  number        = {3},
  pages         = {1--50},
  publisher     = {Association for Computing Machinery (ACM)},
  title         = {Open{QASM} 3: A Broader and Deeper Quantum Assembly Language},
  volume        = {3},
  year          = {2022},
}

@article{Qunity23,
  author        = {Voichick, Finn and Li, Liyi and Rand, Robert and Hicks, Michael},
  title         = {Qunity: A Unified Language for Quantum and Classical Computing},
  year          = {2023},
  issue_date    = {January 2023},
  publisher     = {Association for Computing Machinery},
  address       = {New York, NY, USA},
  volume        = {7},
  number        = {POPL},
  doi           = {10.1145/3571225},
  journal       = {Proc. ACM Program. Lang.},
  month         = jan,
  articleno     = {32},
  numpages      = {31}
}

@article{Sel04,
  title         = {Towards a quantum programming language},
  author        = {Selinger, Peter},
  journal       = {Mathematical Structures in Computer Science},
  volume        = {14},
  number        = {4},
  doi           = {10.1017/s0960129504004256},
  publisher     = {Cambridge University Press (CUP)},
  year          = {2004},
  month         = aug,
  pages         = {527-586}
}

@article{Sho97,
  author        = {Shor, Peter W.},
  title         = {Polynomial-Time Algorithms for Prime Factorization and Discrete Logarithms on a Quantum Computer},
  journal       = {SIAM Journal on Computing},
  volume        = {26},
  number        = {5},
  pages         = {1484--1509},
  year          = {1997},
  doi           = {10.1137/S0097539795293172}
}

@inproceedings{SV05,
  author        = "Selinger, Peter and Valiron, Beno{\^i}t",
  editor        = "Urzyczyn, Pawe{\l}",
  title         = "A Lambda Calculus for Quantum Computation with Classical Control",
  booktitle     = "Typed Lambda Calculi and Applications, {TLCA} 2005",
  year          = {2005},
  publisher     = "Springer Berlin Heidelberg",
  pages         = "354--368",
  doi           = {10.1007/11417170\_26},
  isbn          = "978-3-540-32014-2"
}

@incollection{SV09,
  author        = {Selinger, Peter and Valiron, Beno{\^i}t},
  title         = {Semantic Techniques in Quantum Computation},
  booktitle     = {Quantum lambda calculus},
  editor        = {Gay, Simon and Mackie, Ian},
  publisher     = {Cambridge University Press},
  pages         = {135--172},
  doi           = {10.1017/cbo9781139193313.005},
  year          = {2009}
}

@inproceedings{SVV18,
  author        = "Sabry, Amr and Valiron, Beno{\^i}t and Vizzotto, Juliana Kaizer",
  editor        = "Baier, Christel and Dal Lago, Ugo",
  title         = "From Symmetric Pattern-Matching to Quantum Control",
  booktitle     = "Foundations of Software Science and Computation Structures",
  year          = "2018",
  publisher     = "Springer International Publishing",
  address       = "Cham",
  pages         = "348--364",
  isbn          = "978-3-319-89366-2",
  doi           = {10.1007/978-3-319-89366-2_19}
}

@article{TCM+21,
  title         = {Computational Advantage from the Quantum Superposition of Multiple Temporal Orders of Photonic Gates},
  author        = {Taddei, M\'arcio M. and Cari\~ne, Jaime and Mart\'{\i}nez, Daniel and Garc\'{\i}a, Tania and Guerrero, Nayda and Abbott, Alastair A. and Ara\'ujo, Mateus and Branciard, Cyril and G\'omez, Esteban S. and Walborn, Stephen P. and Aolita, Leandro and Lima, Gustavo},
  journal       = {PRX Quantum},
  volume        = {2},
  issue         = {1},
  pages         = {010320},
  numpages      = {16},
  year          = {2021},
  month         = {Feb},
  publisher     = {American Physical Society},
  doi           = {10.1103/PRXQuantum.2.010320}
}

@inproceedings{Yam01,
  title         = {Confluence and {{Termination}} of {{Simply Typed Term Rewriting Systems}}},
  booktitle     = {Rewriting {{Techniques}} and {{Applications}}},
  author        = {Yamada, Toshiyuki},
  editor        = {Middeldorp, Aart},
  year          = {2001},
  pages         = {338--352},
  publisher     = {Springer},
  address       = {Berlin, Heidelberg},
  doi           = {10.1007/3-540-45127-7_25},
  isbn          = {978-3-540-45127-3},
}

@article{Yam20,
  title        = {A {{Schematic Definition}} of {{Quantum Polynomial Time Computability}}},
  author       = {Yamakami, Tomoyuki},
  year         = {2020},
  month        = dec,
  journal      = {The Journal of Symbolic Logic},
  volume       = {85},
  number       = {4},
  pages        = {1546--1587},
  issn         = {0022-4812, 1943-5886},
  doi          = {10.1017/jsl.2020.45}
}

@book{Yin24,
  title        = {Foundations of quantum programming},
  author       = {Ying, Mingsheng},
  year         = {2024},
  publisher    = {Elsevier},
  doi          = {10.1016/C2014-0-02660-3}
}

@misc{YZ24,
  title        = {Verification of Recursively Defined Quantum Circuits},
  author       = {Mingsheng Ying and Zhicheng Zhang},
  eprint       = {2404.05934},
  year         = {2024}
}

@inproceedings{Silq20,
  series = {PLDI ’20},
  title = {Silq: a high-level quantum language with safe uncomputation and intuitive semantics},
  DOI = {10.1145/3385412.3386007},
  booktitle = {Proceedings of the 41st ACM SIGPLAN Conference on Programming Language Design and Implementation},
  publisher = {ACM},
  author = {Bichsel,  Benjamin and Baader,  Maximilian and Gehr,  Timon and Vechev,  Martin},
  year = {2020},
  month = jun,
  pages = {286–300},
  collection = {PLDI ’20}
}

@article{YVC24,
  title = {Quantum Control Machine: The Limits of Control Flow in Quantum Programming},
  volume = {8},
  ISSN = {2475-1421},
  DOI = {10.1145/3649811},
  number = {OOPSLA1},
  journal = {Proceedings of the ACM on Programming Languages},
  publisher = {Association for Computing Machinery (ACM)},
  author = {Yuan,  Charles and Villanyi,  Agnes and Carbin,  Michael},
  year = {2024},
  month = apr,
  pages = {1–28}
}

@inbook{VEB90,
  title = {Machine Models and Simulations},
  ISBN = {9780444880710},
  DOI = {10.1016/b978-0-444-88071-0.50006-0},
  booktitle = {Algorithms and Complexity},
  series = {Handbook of Theoretical Computer Science},
  publisher = {Elsevier},
  author = {van EMDE BOAS,  Peter},
  year = {1990},
  pages = {1–66}
}

\appendix
\section{Additional Material}
\label{app:additional}

\begin{definition}[Size of a term]
  \label{def:size}
  Given a term $t$, its size, written $\size t \in \mbN$, is defined
  inductively as follows:
  \begingroup
  \renewcommand{\arraystretch}{1.5}
  \[
    \begin{array}{c}
      \size x \triangleq 1
      \qquad
      \size \zket \triangleq 1
      \qquad
      \size \oket \triangleq 1
      \qquad
      \size{\qcasedefault t} \triangleq 1 + \size t + \size{t_0} + \size{t_1} \\
      \size{c(t_1, \dots, t_n)} \triangleq 1 + \sum_{i=1}^n \size{t_i}
      \qquad
      \size{\matchdefaultcollapsed t} \triangleq 1 + \size t +
      \sum_{i=1}^n \size{t_i} \\
      \size{\lbd x t} \triangleq 1 + \size t
      \qquad
      \size{\letrec f x t} \triangleq 1 + \size t
      \qquad
      \size{t_1t_2} \triangleq \size{t_1} + \size{t_2}
      \\
      \size{\sum_{i=1}^n \alpha_i \cdot t_i} \triangleq 1 + \sum_{i=1}^n \size{t_i}
      \qquad
      \size{\shape t} \triangleq 1 + \size t
    \end{array}
  \]
  \endgroup
\end{definition}

\begin{example}[Hybrid function --- \emph{BB84}~\cite{BB84}]
  \label{ex:hybrid-function}
  $\hyrql$ is able to use quantum and classical data simultaneously.
  For example, the well-known protocol \emph{BB84}~\cite{BB84}, creating
  a quantum key from a classical key, can be implemented as follows:
  \[
    \begin{aligned}
      \mathtt{not}&\triangleq \lbd q \qcase q \oket \zket \\
      \mathtt{cc} &\triangleq \lbd b \lbd f \lbd q \match b {
        0 \to q,
        1 \to f\,q
      } \\
      \mathtt{op} &\triangleq \lbd q \lbd n \match n {
        (x, h) \to
        \mathtt{cc}\,h\,\mathtt{Had}(\mathtt{cc}\,x\,\mathtt{not}\,q)
      } \\
      \mathtt{keygen} &\triangleq  \letrec f  l \match l {
        \nil \to \nil,
        h::t \to (\mathtt{op}\,\zket\,h) :: (f\,t)
      }
    \end{aligned}
  \]
  The term $\mathtt{cc}$ applies a gate $f$ conditionally with classical
  control on a bit $b$.
  The term $\mathtt{op}$ reads a pair $(x, h)$, and applies the NOT gate if
  $x = 1$, and then applies the Hadamard gate if $h = 1$. Finally,
  $\mathtt{keygen}$ generates the qubit list key.
\end{example}

\begingroup
\newcommand{\letsug}[3]{\mathtt{let}\,#1=#2\, \mathtt{in}\,#3}
\begin{example}[Quantum Fourier Transform]
  \label{ex:qft}
  We can encode the well-known Quantum Fourier Transform (QFT) in our language,
  with however two points to notice:
  \begin{itemize}
    \item As our natural numbers are defined recursively, defining
      phase gates is not direct.
      We thus assume that we are given a term $\phi$ such that
      $\phi\,n\,\zket = \zket$ and
      $\phi\,n\,\oket = e^{i\pi/2^{n-1}} \oket$.
    \item We will use the syntactic sugar to split 2 qubits and rename them:
      \[
        \letsug{x \otimes y}{s}{t} \triangleq \match s {x \otimes y
        \to t}
      \]
  \end{itemize}
  \[
    \begin{aligned}
      \mathtt{cphase} &\triangleq \lbd x \lbd n \match x {
        c \otimes t \to \qcase c {
          \oket \otimes t
        } {
          \zket \otimes \phi\,n\,t
        }
      } \\
      \mathtt{rot} &\triangleq \letrec f q \lbd l \lbd n \match l {
        &\nil \to q \otimes \nil,
        h::t \to \letsug{h' \otimes q'}{\mathtt{cphase}(h \otimes q)(n)}{
          \\
          &\letsug{q'' \otimes t'}{f(q')(t)(S(n))}{
            q'' \otimes (h'::t')
          }
        }
      } \\
      \mathtt{rec} &\triangleq \letrec f l \match l {
        \nil \to \nil, h::t \to \letsug{h' \otimes
        t'}{\mathtt{rot}(\mathtt{had}\,h)(t)(S(S(0)))}{
          h'::(f\,t')
        }
      } \\
      \mathtt{append} &\triangleq \letrec f l \lbd x \match l {
        \nil \to x :: \nil, h::t \to h::f\,t\,x
      } \\
      \mathtt{swap} &\triangleq \letrec f l \match l {
        \nil \to \nil,
        h::t \to \mathtt{append}(f\,t)\,h
      } \\
      \mathtt{qft} &\triangleq \lbd l \mathtt{swap}(\mathtt{rec}(l))
    \end{aligned}
  \]
\end{example}
\endgroup

\begingroup
\renewcommand{\closedtyping}[2]{;\vdash #1 : #2}

For the sake of simplicity in the following typing trees,
we may denote $\closedtyping t T$ for $\typing{}{} t T$.

\begin{example}

  As $\zket, \oket$ are pure values, then we get $\zket \perp \oket$
  directly from the definition.
  This allows us to type $\ket \pm$:

  \[
    \ptsc{
      \infer0[(ax$_0$)]{\closedtyping \zket \qbit}
      \infer0[(ax$_1$)]{\closedtyping \oket \qbit}
      \hypo{\size{\frac 1 {\sqrt 2}}^2 + \size{\pm \frac 1 {\sqrt 2}}^2 = 1}
      \hypo{\zket \perp \oket}
      \infer4[(sup)]{\closedtyping{\ket \pm}{\qbit}}
    }
  \]
  It is also easy to check that $\ket + \perp \ket -$:
  \newcommand{\stwo}{\frac 1 {\sqrt 2}}
  \[
    \begin{aligned}
      \innerprod{\ket +}{\ket -} &= \stwo \stwo \kron{\zket}{\zket} -
      \stwo \stwo \kron{\zket}{\oket} + \stwo \stwo
      \kron{\oket}{\zket} - \stwo \stwo \kron{\oket}{\oket} \\
      &= \frac 1 2 - \frac 1 2 = 0
    \end{aligned}
  \]

  we can also type the Hadamard gate as follows:

  \[
    \ptsc{
      \infer0{\typing{}{x : \qbit} x \qbit}
      \hypo{\closedtyping {\ket +} \qbit}
      \hypo{\closedtyping {\ket -} \qbit}
      \hypo{\ket + \perp \ket -}
      \infer4[(qcase)]{\typing{}{x : \qbit}{\qcase x {\ket
      +}{\ket -}}{\qbit}}
      \infer1[(lbd)]{\closedtyping{\mathtt{Had}}{\qbit \linfunc \qbit}}
    }
  \]
\end{example}

\begin{example}
  The quantum switch from \autoref{ex:qs-syntax} can be typed by our language.
  For conciseness, we write $v = \zket \otimes f(g\,t)$, $w = \oket
  \otimes g(f\,t)$,
  $\Gamma = f : \qbit \linfunc \qbit, g : \qbit \linfunc \qbit$,
  $s = \match q { (c \otimes t) \to \qcase c v w}$,
  $\Delta = c : \qbit, t: \qbit$ and $Q_2 = \qbit \otimes \qbit$.
  The orthogonality predicate $v \perp w$ is verified, as for any
  substitution of
  $f, g, t$, $\sigma(v) \sreduces \zket \otimes v'$ and
  $\sigma(w) \sreduces \oket \otimes w'$. Both terms have the same shape $\vunit \otimes \vunit$ and their inner product is zero.
  Note that we substitute by terminating terms, thus $f, g$ terminate and we
  can evaluate the inner product.

  \[
    \ptsc{
      \hypo{\typing \Gamma {c : \qbit} c \qbit}
      \hypo{\typing{\Gamma}{t : \qbit}{v}{Q_2}}
      \hypo{\typing{\Gamma}{t : \qbit}{w}{Q_2}}
      \hypo{v \perp w}
      \infer4{\typing \Gamma \Delta {\qcase c v w} Q_2}
      \hypo{\typing \Gamma {q: Q_2} q Q_2}
      \infer2{\typing \Gamma {q : Q_2}{s}{Q_2}}
      \infer1{\typing \Gamma {}{\lbd q s}{Q_2 \linfunc Q_2}}
      \infer1{\typing {f : Q_2 \linfunc Q_2} {} {\lbd g \lbd q
      s}{(\qbit \linfunc \qbit) \nonlinfunc (Q_2 \linfunc Q_2)}}
      \infer1{\closedtyping {\mathtt{QS}}{(\qbit \linfunc \qbit)
      \nonlinfunc (\qbit \linfunc \qbit) \nonlinfunc (Q_2 \linfunc Q_2)}}
    }
  \]
\end{example}

\newcommand{\lqbit}{\typelist \qbit}
\newcommand{\slqbit}{\typelist \unit}
\newcommand{\len}{\mathtt{len}_\unit}
\begin{example}
  \label{ex:len-app}
  Coming back on \autoref{ex:len-typing},
  the following typing tree derives
  typing for $\mathtt{len}$, for a given type
  $A \in \bcset$, with the following notations
  to shorten the tree:
  $\Gamma_f = \vartyping f \typelist A \nonlinfunc \nat$, $\Gamma_x =
  \vartyping x {\typelist A}$, and $\Gamma = \Gamma_f, \Gamma_x, \vartyping h
  A, \vartyping t \typelist A$. For conciseness, we also do not type $f$ and
  $t$ on the right part of the tree, which can both be typed as
  $\typing \Gamma {} f {\typelist A \nonlinfunc \nat}$ and
  $\typing \Gamma {} t \typelist A$, using (ax$_c$) rules. We also drop the 
  signature on constructors once they are fully applied.

  \[
    \begin{prooftree}[small]
      \infer0[(ax)]{\typing{\Gamma_f, \Gamma_x}{}{x}{\typelist A}}
      \infer0[(cons)]{\typing{\Gamma_f, \Gamma_x}{}{0^\nat}{\nat}}
      \infer0[(app$_c$)]{\typing \Gamma {} {ft}{\nat}}
      \infer1[(cons)]{\typing \Gamma {} {S^{\sign \nat \nat}(ft)}{\nat}}
      \infer3[(match)]{\typing{\Gamma_f, \Gamma_x}{}{\match x {\relax
      \nil \to 0 \,, h :: t \to S(ft)}} \nat}
      \infer1[(abs$_c$)]{\typing{\Gamma_f}{}{\lbd x {\match x {\relax
      \nil \to 0 \,, h :: t \to S(ft)}}}{\typelist A \nonlinfunc \nat}}
      \infer1[(rec)]{\closedtyping{\letrec f x {\match x {\relax \nil
      \to 0 \,, h :: t \to S(ft)}}}{\typelist A \nonlinfunc \nat}}
    \end{prooftree}
  \]

  In particular, as $h$ is not present in the syntax of $S(f\,t)$,
  we could not type $\mathtt{len}$ if $h$ must be used linearly, thus
  if $A \in \qvar$.
\end{example}

\begin{example}
  Assuming we have a bit type $\mathtt{bit}$ for both constructors $0 ::
  \mathtt{bit}$ and $1 :: \mathtt{bit}$, the typing of
  \autoref{ex:hybrid-function}
  can be derived, by combining most of the techniques exhibited from the
  previous examples, giving us $\closedtyping {\mathtt{keygen}}
  {\typelist{\mathtt{bit}}\nonlinfunc \typelist{\qbit}}$. Note that when we used
  programs in the syntax of other terms, such as $\mathtt{Had}$, it is just
  notation to avoid rewriting the whole term.
\end{example}
\endgroup

\section{Proofs for Section~\ref{subsec:properties}}

\subsection{Intermediate results}

In this subsection, we state some intermediate results
and definitions that we will use through the main proofs.
Furthermore, we may use the notation $\{x \to v\}$ for a
substitution $\{ v / x \}$, and define the \emph{support} of
a substitution $\subsupp \sigma$ as the set of variables $x_i$
such that there exists $t_i$ with $\{ t_i / x_i \} \in \sigma$.
When we do not need precision on the linearity of a functional term,
we may use $\anyfunc \triangleq \linfunc \cup \nonlinfunc$.
We also denote the free variables of $t$ by $\fv t$.

\weakening*

\begin{proof}
  This can be proven by induction on $t$, as the axiom rules satisfy
  this weakening property.
  Furthermore, the only variable removed from the non-linear context
  are bound variables,
  thus $\Gamma'$ stays untouched at each point. \qedhere
\end{proof}

\begin{lemma}\label{lem:functype}
  Let $\tvar \in \mathcal T$.  Then, there exists a number
  $n \in \mathbb N$, types $T_i$ for $1 \leq i \leq n$,
  and a base type $\bkvar$ such that
  $\tvar = \tvar_1 \anyfunc \dots \anyfunc \tvar_n \anyfunc \bkvar$.
\end{lemma}

\begin{proof}
  Direct by induction on $\tvar$: either it is a base type itself, thus we take
  $n = 0$; or $\tvar = \tvar_1 \anyfunc \tvar'$, then we use the induction
  hypothesis on $\tvar'$ and right associativity to conclude. \qedhere
\end{proof}

We also introduce an operator $\qt p t$, called the \emph{quantity},
which aims to measure the amplitude of a pure term p in a term $t$,
\eg{} a superposition.

\begin{definition}\label{def:quantity}
  Let $p$ be a pure term.
  We define the following map $\theta_p : \Lambda \to \mathbb C$ as follows:
  \begingroup
  \renewcommand{\arraystretch}{1.5}
  \[
    \begin{array}{c}
      \qt p x \triangleq \delta_{p, x} \qquad \qt p \zket \triangleq
      \delta_{p, \zket}
      \qquad \qt p \oket \triangleq \delta_{p, \oket} \\
      \qt p {\qcase {t'} {t_0} {t_1}} \triangleq \delta_{p,
      \qcase{p'}{t_0}{t_1}} \qt{p'}{t'} \\
      \qt p {c(t_1, \dots, t_n)} \triangleq \delta_{p, c(p_1, \dots, p_n)}
      \times_{i=1}^n \qt{p_i}{t_i} \\
      \qt p {\matchdefaultcollapsed{t'}} \triangleq
      \delta_{p, \matchdefaultcollapsed{p'}} \qt{p'}{t'} \\
      \qt p {\lbd x t} \triangleq \delta_{p, \lbd x t} \qquad \qt p {\letrec f
      x t} \triangleq \delta_{p, \letrec f x t} \qquad
      \qt p {t_1t_2} \triangleq \kron p {p_1p_2} \kroneq {p_1} {t_1} \kroneq
      {p_2} {t_2}  \\
      \qt p {\sum_{i=1}^n \alpha_i \cdot t_i} \triangleq
      \sum_{i=1}^n \alpha_i
      \cdot \qt p {t_i}
      \qquad \qt p {\shape{t'}} \triangleq \kron p {\shape{p'}}
      \kroneq {p'} {t'}
    \end{array}
  \]
  \endgroup
  We have used the notation $\kron p t$ for the Kronecker symbol,
  and $\kroneq p t$ for the equivalence Kronecker symbol, meaning
  $\kroneq p t = 1$ iff $p \equiv t$.
\end{definition}

\begin{lemma}\label{lem:quantity-one}
  Let $p$ be a pure term.
  Then $\qt p p = 1$.
\end{lemma}

\begin{proof}
  Proven by induction on $p$, direct by definition of the quantity. \qedhere
\end{proof}

\begin{lemma}\label{lem:quantity-equiv}
  Let $p$ be a pure term and $t, t'$ be two terms of the language.
  If $t \equiv t'$, then $\qt p t = \qt p {t'}$.
\end{lemma}

\begin{proof}
  By induction on each rule defining $\equiv$ in \autoref{tab:equiv};
  the transitive and reflexive case can be derived from this.
  The first two lines of rules are verified directly as $\mathbb C$
  is a vectorial space.
  The three linearity rules also work by definition of summation and
  the linear constructs.
  Finally, we can prove the $\equivcontext$ case by doing an
  induction on the depth of $\equivcontext$,
  the base case being just $t \equiv t'$ and then obtaining the
  result by induction. \qedhere
\end{proof}

\begin{lemma}\label{lem:quantity-pure}
  Let $s, t$ be two pure terms.
  Then $\qt s t = \kroneq s t$.
  Furthemore, if $s, t$ are values, then $\qt s t = \kron s t$.
\end{lemma}

\begin{proof}
  Note by \autoref{lem:quantity-one}, $\qt s s = 1$.
  If $\qt s t = 0$, then $s \not\equiv t$, else
  \autoref{lem:quantity-equiv} would fail.
  Now, one can check by induction on $t$ that if $\qt s t = 1$, then
  $s \equiv t$.
  This is done directly, mainly proving $\equiv$ using $\equivcontext$.
  The second part of the lemma is also obtained directly by induction
  of $t$, by looking at the definition of the quantity. \qedhere
\end{proof}

\begin{restatable}{lemma}{canonical}
  \label{lem:canonical}
  Let $t$ be a term of the language. Suppose $t \not\equiv 0 \cdot t'$.
  Then $t$ has a canonical form $\sum_{i=1}^n \alpha_i \cdot t_i$,
  that is unique up to reordering and equivalence on the $t_i$.
\end{restatable}

\begin{proof}
  Let us prove existence by induction on $t$.

  \begin{itemize}
    \item If $t = x, \zket, \oket, \lbd x s, \letrec f x s, s_1 s_2$ or $\shape
      s$, then $t$ is directly pure and $t \equiv 1 \cdot t$, which is a
      canonical form.
    \item Suppose $t = c(t_1, \dots, t_n)$. If we have any $t_i
      \equiv 0 \cdot t_i'$, then $t \equiv 0 \cdot c(t_1, \dots,
      t_i', \dots, t_n)$.
      Therefore, we can apply the induction hypothesis on each $t_i$,
      meaning we have $t_i \equiv \sum_{j_i = 1}^{n_i} \alpha_{j_i}
      \cdot t_{j_i}$.
      By developing $t$ under $\equiv$, we have $t \equiv
      \sum_{j_1=1}^{n_1} \dots \sum_{j_n=1}^{n_n} \alpha_{j_1} \dots
      \alpha_{j_n} \cdot c(t_{j_1}, \dots, t_{j_n})$.
      This can be seen as one big sum $\sum_{k=1}^{n_1 \dots n_k}
      \alpha_k \cdot c(t_1^k, \dots, t_n^k)$, where $\alpha_k =
      \alpha_{j_1} \dots \alpha_{j_n}$.
      In particular, $\alpha_k \neq 0$ as each $\alpha_{j_i} \neq 0$;
      and if $c(t_1^k, \dots, t_n^k) \equiv c(t_1^{k'}, \dots,
      t_n^{k'})$, then by \autoref{lem:quantity-equiv},
      $\qt {c(t_1^k, \dots, t_n^k)} {c(t_1^{k'}, \dots, t_n^{k'})} =
      1$ as they are both pure,
      but by analyzing the formula, this implies $\qt {t_i^k}
      {t_i^{k'}} = 1$, thus by \autoref{lem:quantity-equiv}, $t_i^k
      \equiv t_i^{k'}$.
      By induction hypothesis, this is not possible for all $1 \leq i
      \leq n$, therefore terms are pairwise not equivalent, which concludes.
    \item The same process can be done for $t = \qcasedefault{t'}$ or $t
      = \match t \dots$, as they have an identical behaviour to a
      constructor with one slot.
    \item Finally, suppose $t = \sum_{i=1}^n \alpha_i \cdot t_i$.
      Any $t_i \equiv 0 \cdot t_i'$ will be removed from the sum
      through $t + 0 \cdot t' \equiv t$, then we can apply the i.h.
      on the remaining $t_i$.
      Finally, using distributivity, we can obtain a big sum
      $\sum_{k=1}^n \beta_k \cdot s_k$, where $s_k$ are pure terms.
      We can group $s_k$ that are equivalent, then remove any
      potential term with a $0$ phase, and the term obtained
      satisfies the canonical form definition.
      Note that as $t \not\equiv 0 \cdot t'$, the sum obtained is
      actually well defined, in particular there is at least one term.
  \end{itemize}

  Unicity comes from our quantity function defined above.
  Suppose $t$ has two canonical forms $t \equiv \sum_{i=1}^n \alpha_i
  \cdot t_i \equiv \sum_{j=1}^m \beta_j \cdot s_j$.
  As all $t_i$ are pure, for any $1 \leq i_0 \leq n$, $\qt {t_{i_0}}
  {\sum_{i=1}^n \alpha_i \cdot t_i} = \sum_{i=1}^n \alpha_i \qt{t_{i_0}} {t_i}$.
  By \autoref{lem:quantity-pure}, $\qt{t_{i_0}} {t_i} = \kroneq
  {t_{i_0}} {t_i}$, however the $t_i$ are not equivalent pairwise,
  thus $\qt {t_{i_0}} {\sum_{i=1}^n \alpha_i \cdot t_i} = \alpha_{i_0} \neq 0$.
  However, by the same argument, $\qt{t_{i_0}}{\sum_{j=1}^m \beta_j
  \cdot s_j} = \sum_{j=1}^m \beta_j \qt{t_{i_0}} {s_j}$.
  If there is no $s_j \equiv t_{i_0}$, then
  \autoref{lem:quantity-equiv} fails, thus there is at least one;
  and there cannot be more than one, else, $s_j \equiv t_{i_0} \equiv
  s_{j'}$, thus by transitivity $s_j \equiv s_{j'}$, which is not
  possible as this is a canonical form.
  Therefore, there exists $j_0$ such that $\qt{t_{i_0}}{\sum_{j=1}^m
  \beta_j \cdot s_j} = \beta_{j_0}$;
  by \autoref{lem:quantity-equiv}, we obtain $\alpha_{i_0} = \beta_{j_0}$.
  Thus, any $t_i$ has a unique and distinct (as they are not
  equivalent) match in $s_j$, therefore $n \leq m$.
  The same process can be done the other way around, so $m \leq n$,
  thus $ m = n$.
  So each sum contains the same number of elements, such that there
  is a map from $1, \dots, n$ to $1, \dots, n$ that associates $i$
  with $j$ such that $t_i \equiv s_j$ and $\alpha_i = \beta_j$.
  Therefore, both sums are equal up to an equivalence on $t_i$ and
  permutation of the $t_i$. \qedhere
\end{proof}

\begin{lemma}\label{lem:canonical-shape}
  Let $t$ be a well-typed term, suppose that $t$ terminates
  and has a canonical form $\sum_{i=1}^n \alpha_i \cdot s_i$. Then
  any element has the same
  shape, meaning $\shape{s_i} \sreduces v_i$,
  $\shape{s_j} \sreduces v_j$, and $v_i = v_j$ for any $i, j$.
\end{lemma}

\begin{proof}
  By induction on the typing, and the constructed canonical form. Few
  points to note:

  - for (qcase), as $t_0 \perp t_1$, in particular they have the same shape,
  thus any $s_i$ will reduce to a superposition of $t_0, t_1$;

  - for (match), even if the inside argument is the superposition, by
  induction hypothesis,
  all elements have the same shape, thus in particular reduce to the same $t_i$.

  - For (sup), $t_i \equiv t_j$, thus they have the same shape. \qedhere
\end{proof}

Another process that we will do in proofs is completing canonical forms.
Given two terms $s, t$ with canonical forms $\sum_{i=1}^n \alpha_i \cdot s_i$
and $\sum_{j=1}^m \beta_j \cdot t_j$, one can always create the set $S$,
that contains all the $s_i$, $t_j$ such that they are pairwise non equivalent.
Then, by adding $0 \cdot t'$, we can write each canonical form as
$\sum_{k=1}^l \alpha_l \cdot w_l$ and $\sum_{k=1}^l \beta_l \cdot w_l$.
This is not a canonical form as some phases are zero, but it is equivalent to
the first canonical form, thus to $s, t$; and such shape is actually
close enough
to keep most existing results on the canonical form, such as normalization.

The following lemma abuses the operational semantics, by supposing that we
have a rule allowing to reduce any term $\sum_{i=1}^n 0 \cdot t_i$ to
$\sum_{i=1}^n 0 \cdot t_i'$ when $t_i \zoreduces t_i'$. Our type system
can easily be expanded to satisfy the properties; in fact, all the following
results, \ie{} subject reduction and confluence are proven by taking this case
into account. Furthermore, we will later prove that any well-typed term
cannot be expressed as so, thus this rule is only necessary for the proofs,
and we have thus chosen to omit it.

\begin{lemma}
  \label{lem:reduces-sum}
  Let $t = \sum_{i=1}^n \alpha_i \cdot t_i$, suppose $t_i \sreduces v_i$.
  Then $t \sreduces \sum_{i=1}^n \alpha_i \cdot v_i$.
\end{lemma}

\begin{proof}
  For each $t_i$ such that $t_i \equiv 0 \cdot s$,
  or such that $\alpha_i = 0$, we can remove it via (equiv),
  thus $t \equiv \sum_{k=1}^l \alpha_k \cdot t_k$.
  By writing each canonical form, up to completion, as
  $t_k \equiv \sum_{j=1}^m \beta_{k,j} \cdot s_j$,
  $s_j \zoreduces s_j'$
  and $\sum_{k=1}^l \alpha_k \cdot \beta_{k,j} = \gamma_j$:

  \[
    t \equiv \sum_{j=1}^m \gamma_j \cdot s_j \equiv \sum_{j, \gamma_j
    \neq 0} \gamma_j \cdot s_j
    \reduces \sum_{j, \gamma_j \neq 0} \gamma_j \cdot s_j' \equiv
    \sum_{j=1}^m \gamma_j \cdot s_j' \equiv \sum_{k=1}^n \alpha_k \cdot t_k',
  \]

  as the $s_j$ are pairwise not equivalent, either all $\gamma_j = 0$, or
  it is a canonical form. If the obtained form is a value,
  then we can stop here, and the result will be correct; even if $t_i$ is not
  a value, the non-value parts will still be canceled via (Equiv) at the
  end, as the reduction is deterministic.
  Else, we can appy (Can), and we can replace $0 \cdot s_j$ by $0 \cdot s_j'$.
  Furthermore, we can add back the removed terms from the beginning, adding
  back $t_i'$ instead of $t_i$, to get that
  $\sum_{k=1}^l \alpha_k \cdot t_k' \equiv \sum_{i=1}^n \alpha_i \cdot t_i'$.
  We conclude by (Equiv), and chaining this multiple times;
  such process terminates as the $t_i$ do terminate. \qedhere
\end{proof}

\subsection{Proofs of confluence}

\begin{lemma}[Equivalence preservation]
  \label{lem:confluence-equiv}
  Let $s_1, s_2$ be two well-typed terms.
  Suppose $s_i \reduces t_i$ for $i = 1, 2$.
  Then $t_1 \equiv t_2$.
\end{lemma}

\begin{proof}
  We prove this result by induction of the sum
  of the depth of the derivation trees of $s_i \reduces t_i$.

  Consider the base case where both derivation trees have the
  same root. Apart from (Can) and (Shape$_s$), having the
  same root implies that $s_1 = s_2$, and as reduction is
  deterministic, $t_1 = t_2$, thus $t_1 \equiv t_2$.
  In the case of (Can), by \autoref{lem:canonical}, if one has
  a canonical form, then by $\equiv$ the second also has a canonical
  form, which is identical up to commutation and equivalence on the terms.
  Thus, $s_1 = \sum_{j=1}^n \alpha_i \cdot r_i$, and
  $s_2 = \sum_{j=1}^n \beta_i \cdot q_i$, where both sums are identical,
  up to commutation and each $q_i$ is equivalent to a $r_j$.
  We can apply the induction hypothesis on $q_i, r_j$, to get that the
  reduced terms are also equivalent, and up to commuting the terms back,
  we obtain the equivalence at the end. If not, both terms have
  $0$ phase, and equivalence is direct from this.
  The same process can be applied for (Shape$_s$), but inside the
  \textshape{} construct, by \autoref{lem:canonical-shape}, as it is a value.

  Now, suppose that the trees do not have the same root.
  First, do not consider any (Shape) rule. Again, apart from (Can) and (Equiv),
  all left terms are pure, and by \autoref{lem:quantity-equiv},
  no distinct rule can yield two equivalent terms.
  Suppose the first tree has a (Equiv) root, meaning:

  \[
    \begin{prooftree}
      \hypo{s_1 \equiv q_1}
      \hypo{q_1 \reduces r_1}
      \hypo{r_1 \equiv t_1}
      \infer3[(Equiv)]{s_1 \reduces t_1}
    \end{prooftree}
  \]

  By transitivity, $q_1 \equiv s_2$, thus we can apply the induction hypothesis
  on $q_1 \reduces r_1$ and $s_2 \reduces t_2$ to get $r_1 \equiv t_2$,
  and by transitivity, $t_1 \equiv t_2$.
  Same can be done for the right term.

  Now, suppose the first tree has a (Can) root. As the right tree has neither a
  (Can) or (Equiv) root, it has a root with a a pure term. This implies that
  all phases are not zero, and that $s_1$ is a canonical form. By uniqueness,
  $s_1 = 1 \cdot q_1$, with $q_1 \equiv s_2$, and $t_1 = 1 \cdot r_1$.
  We can apply the induction hypothesis on $q_1 \reduces r_1$,
  $s_2 \reduces t_2$, to get $t_1 \equiv t_2$ by transitivity of (Equiv).
  Same can be done for the right term.

  The same process can be done for the (Shape) rules, but inside the
  \textshape{} construct. \qedhere
\end{proof}

\confluence*

\begin{proof}
  Let us prove semi-confluence, which implies confluence.
  Let $t \reduces t_1$ and $t \sreduces t_2$.
  If $t = t_2$, then result is direct; else $t \reduces t_3 \sreduces t_2$.
  If $t_3 = t_2$, by the above result, $t_1 \equiv t_2$, which concludes.
  Else $t \reduces t_3 \reduces t_4 \sreduces t_2$.
  By the above result, $t_3 \equiv t_1$.
  And as $t_3 \reduces t_4$, then $t \reduces t_4$ by (Equiv).
  Therefore, $t \reduces t_1 \reduces t_4 \sreduces t_2$, and we
  conclude. \qedhere
\end{proof}

\subsection{Proofs for Canonical Forms}

\begin{lemma}\label{lem:quantity-nonzero}
  Let $t$ be a term and $s$ be a pure term.
  Suppose $\qt s t \neq 0$.
  Then $t$ has a unique canonical form.
\end{lemma}

\begin{proof}
  Suppose $t \equiv 0 \cdot t'$.
  Then $\qt s t = \qt s {0 \cdot t'} = 0$ by
  \autoref{lem:quantity-equiv} and by definition of
  the quantity; this contradicts the initial
  hypothesis. We can then conclude by \autoref{lem:canonical}. \qedhere
\end{proof}

\canonicaltyped*

\begin{proof}
  Let us first prove that for any well-typed $t$, there exists a pure
  term $s$ such that $\qt st \neq 0$;
  we then conclude by \autoref{lem:quantity-nonzero}.
  By induction on $t$:

  \begin{itemize}

    \item If $t$ is pure, then take $s = t$ and result is direct;

    \item If the root is (cons), meaning $t = c(t_1, \dots, t_n)$,
      then by induction hypothesis,
      as $t_i$ is typed, there is a pure term $s_i$ such that $\qt
      {s_i} {t_i} \neq 0$.
      We can take $s = c(s_1, \dots, s_n)$ and conclude.

    \item Same is done for (qcase) and (match).

    \item If the root is (equiv), $t \equiv t'$, we obtain the result
      by induction hypothesis on $t'$ and conclude by
      \autoref{lem:quantity-equiv}.

    \item If the root is (contr), then $t = \sigma(t')$; we can apply
      the induction hypothesis on $t'$ and take $s = \sigma(s')$
      which will preserve $\sigma$.

    \item Finally, if the root is (sup), $t = \sum_{i=1}^n \alpha_i \cdot t_i$.
      By typing, $t_i$ is well-typed, thus by induction hypothesis and
      \autoref{lem:quantity-nonzero}, $t_i$ has a canonical form.
      Furthermore, as $t_i \perp t_j$, each $t_i$ terminates,
      with $t_i \sreduces \sum_{j=1}^{n_i} \beta_{i,j} \cdot w_j^i$,
      being canonical forms. One can complement each canonical
      form to get $t_i \sreduces \sum_{j=1}^m \beta_{i,j} w_j$,
      where the $w_j$ are pure and distinct pairwise.
      We can thus use \autoref{lem:reduces-sum}, and get that
      $t \sreduces \sum_{j=1}^m (\sum_{i=1}^n \alpha_i \beta_{i,j}) \cdot w_j$.
      If we suppose that $t \equiv 0 \cdot t'$, taking $t' = v$,
      by \autoref{thm:confluence}, $0 \cdot v \equiv \sum_{j=1}^m
      (\sum_{i=1}^n \alpha_i \beta_{i,j}) \cdot w_j$.
      By \autoref{lem:quantity-equiv}, $\qt {w_j}{0 \cdot v} = 0 =
      \sum_{i=1}^n \alpha_i \beta_{i,j}$.
      One can check that by definition of $t_i \perp t_j$, the rows
      $(\beta_{i,1}, \dots, \beta_{i,n})$
      are orthogonal, thus this would imply that $\alpha_i = 0$,
      which contradicts
      the hypothesis of (sup). Therefore, $t \not\equiv 0 \cdot t'$,
      by \autoref{lem:canonical}, $t$ has a canonical form
      $\sum_{k=1}^l \gamma_k \cdot u_k$,
      and take $s = u_1$ to conclude.
  \end{itemize}

  Finally, one can remark that by construction of the canonical form in
  \autoref{lem:canonical}, along with an induction on typing, the pure
  terms obtained are well-typed; such result is possible as the canonical
  form is unique, thus the one that we exhibited will match.
\end{proof}

\subsection{Proofs of progress}

\begin{lemma}\label{lem:equiv-pure}
  Let $t$ be a well-typed pure term of type $\bkvar$, and $v$ be a value.
  We suppose that $t \equiv v$, then $t$ is a value.
\end{lemma}

\begin{proof}
  As $t$ is pure, $ \qt t v = \qt t t = 1$.
  By definition of the quantity, it is easy to check
  that $\qt t v = 1$ implies that $t$ is a value too,
  thus one can conclude. \qedhere
\end{proof}

\begin{lemma}\label{lem:canonical-higher}
  Let $t$ be a well-typed term of type $T_1 \anyfunc T_2$.
  Then its canonical form is $1 \cdot t'$.
\end{lemma}

\begin{proof}
  Direct by the way canonical forms are generated in \autoref{lem:canonical},
  and by typing of the different constructs. \qedhere
\end{proof}

\progress*

\begin{proof}
  We prove this result by induction on the typing of $t$.
  First, let us consider that $t$ is a pure term.

  \begin{itemize}
    \item The rules (ax), (ax$_c$) do not derive a closed term.
    \item The rules (ax$_0$), (ax$_1$) yield directly a value.
    \item If the root is (qcase), then $t = \qcase s {t_0}
      {t_1}$, with $\closedtyping{s}{\qbit}$.
      Either $s \reduces s'$ and $t \reduces \qcase {s'} {t_0}
      {t_1}$ through (\evalcontext);
      or $s \equiv v$, which by \autoref{lem:equiv-pure} indicates
      that $s$ is a pure value, thus $s = \ket i$.
      In any case, we can reduce through (Qcase$_i$) for $i = 0, 1$.
    \item Suppose $t$ is typed with (cons), namely
      $\closedtyping{c(t_1, \dots, t_n)}{\bvar}$.
      This implies $\closedtyping{t_i}{B_i}$, thus we can apply the
      induction hypothesis on $t_i$.
      If each $t_i \equiv v_i$, then $t \equiv c(v_1, \dots, v_n)$
      through \equivcontext, which is a value.
      Else, we take the first $t_i$, starting from the right, such
      that $t_i \reduces t_i'$, and then $t \reduces c(t_1, \dots,
      t_i',\dots, t_n)$ by (\evalcontext).
    \item If the root is (match), then $t = \matchdefault s$,
      with $\closedtyping{s}{\bvar}$.
      Either $s \reduces s'$ and $t \reduces \match {s'} \dots$
      through (\evalcontext);
      or $s \equiv v$, which by \autoref{lem:equiv-pure} indicates
      that $s$ is a pure value, thus $s = \bar c(v_1, \dots, v_m)$.
      As $s$ is of type $B$, then $\bar c \in \cons B$, thus $\bar c
      = c_i$ and $t \reduces t_i$ through (Match).
    \item The rules (abs), (abs$_c$) and (rec) yield directly a value.
    \item Suppose the root of the derivation is (app), meaning $t = t_1 t_2$.
      By typing, $t_1, t_2$ are also closed terms and the induction
      hypothesis can be applied.
      If $t_2 \reduces t_2'$, then $t \reduces t_1t_2'$ by (\evalcontext).
      Else, $t_2 \equiv v_2$.
      If $t_1 \reduces t_1'$, then $t_1v_2 \reduces t_1'v_2$ by
      (\evalcontext), and as $t \equiv t_1v_2$, then $t \reduces
      t_1'v_2$ by (Equiv).
      Else, $t_1 \equiv v_1$.
      By \autoref{lem:canonical-higher}, as $t_1$ has a functional
      type, then its canonical form
      consists of one pure term; as $v_1$ has the same canonical form
      by unicity, this term is a value too.
      Thus, $t_1 \equiv 1 \cdot v \equiv v$.
      Now, the only closed, well-typed values of type $C \nonlinfunc T$ or
      $T \linfunc T'$ are either $\lbd x t'$, and as $(\lbd x t')v_2
      \reduces \sigma(t')$
      via (Lbd), then $t \reduces \sigma(t')$ via (Equiv);
      or $\letrec f x t'$ and as $(\letrec f x t')v_2 \reduces \sigma(t')$
      via (Rec), then $t \reduces \sigma(t')$ via (Equiv).
    \item The rule (sup) does not derive a pure term.
    \item Suppose the root is (shape), meaning $t = \shape s$, with
      $\closedtyping{s}{\bkvar}$.
      Note that $s$ is actually closed because its linear and
      non-linear context are joined in the non-linear context of $t$
      which is empty.
      Then either $s \reduces s'$ and $t \reduces \shape{s'}$ via
      (\evalcontext); or $s \equiv v$.
      Either $v = \ket i$, and $t$ reduces through (Shape$_i$) and (Equiv);
      or $v = c(v_1, \dots, v_n)$ and $t$ reduces through (Shape$_c$)
      and (Equiv);
      or $v$ is a summation, as it is well-typed it has a canonical
      form $w$, and by transitivity $t \equiv w$, thus $t$ reduces
      through (Shape$_s$) and (Equiv).
      Any other value cannot be typed with a type $\bkvar$.
    \item Suppose the root of the derivation is (equiv), namely we
      have $t' \equiv t$.
      By hypothesis, $t$ is closed and thus $t'$ is too.
      Therefore, either $t' \equiv v$ and by transitivity of
      $\equiv$, $t \equiv v$; else $t' \reduces t''$ and then $t
      \reduces t''$ by (Equiv).
    \item The rule (contr) does not derive a closed term.
    \item Finally, the rule (contr) cannot yield a closed term.
  \end{itemize}

  Now, let $t$ be a well-typed closed term; let us consider its
  canonical form $t' = \sum_{i=1}^n \alpha_i \cdot t_i$.
  If $t'$ is a value, then $t$ is equivalent to a value, and we conclude.
  Else, any term $t_i$ verifies $t_i \zoreduces t_i'$ for a given
  $t_i'$, by definition, and thus $t' \reduces \sum_{i=1}^n \alpha_i
  \cdot t_i'$ by (Can), and thus $t \reduces \sum_{i=1}^n \alpha_i
  \cdot t_i'$ by (Equiv), which concludes. \qedhere
\end{proof}

\subsection{Proofs for Subject Reduction}

\begin{lemma}[Substitution lemma]
  The two following properties hold:

  - If $\typing {\Gamma, \vartyping x T'} \Delta t T$, $\typing
  \Gamma \varnothing {t'} T'$ and $\sigma = \{x \subarrow t'\}$, then
  $\typing \Gamma \Delta {\sigma(t)} T$.

  - If $\typing {\Gamma} {\Delta, \vartyping x T'} t T$, $\typing
  \Gamma {\Delta'} {t'} T'$ and $\sigma = \{x \subarrow t'\}$, then
  $\typing \Gamma {\Delta, \Delta'} {\sigma(t)} T$.
\end{lemma}

\begin{proof}
  We prove this by induction on the typing tree of $t$.
  For the first property:
  \begin{itemize}
    \item Suppose the root is (ax), meaning $\typing {\Gamma,
      \vartyping x T'} {\vartyping y T} y T$, then one can check that
      we can also type directly $\typing \Gamma {\vartyping y T} y
      T$, and by definition $\sigma(y) = y$ as $y \notin \subsupp
      \sigma$, which concludes.
    \item Suppose the root is (ax$_c$), then either $t = y$ with $y
      \neq x$, and the same process can be done as above.
      Else, we have $\sigma(t) = \sigma(x) = t'$, which concludes by
      the typing of $t'$.
    \item If the rules are (ax$_0$) or (ax$_1$), then again
      $\sigma(t) = t$ and we can derive the typing of $t$ with
      $\Gamma$ instead of $\Gamma, \vartyping x T'$.
    \item Suppose the root is (qcase), thus $t = \qcasedefault{t'}$.
      By induction hypothesis, $x$ is removed from the context of
      $t', t_0, t_1$, thus we conclude by typing
      $\sigma(t) = \qcase{\sigma(t')}{\sigma(t_0)}{\sigma(t_1)}$.
    \item Suppose the rule is (cons), meaning $\typing {\Gamma,
      \vartyping x T'} {\Delta_i} {t_i} A_i$; by i.h. $\typing \Gamma
      {\Delta_i} {\sigma(t_i)} A_i$, and thus we can type $\typing
      \Gamma {\Delta} {c(\sigma(t_1), \dots, \sigma(t_n))} \bvar$,
      which concludes.
    \item Suppose the rule is (match), the same process can be done
      as for (qcase), by remarking that $x$ cannot be a bound
      variable of a $t_i$, else typing could not happen.
    \item  Suppose the root is (abs), thus we have $\typing {\Gamma,
      \vartyping y T''} {\Delta, \vartyping x T} t T'$; by i.h.
      $\typing \Gamma {\Delta, \vartyping x T} {\sigma(t)} T'$, thus
      by (abs) we can type $\typing \Gamma \Delta {\lbd x \sigma(t)}
      T \linfunc T'$, which concludes as $\lbd x \sigma(t) =
      \sigma(\lbd x t)$ as $x \notin \subsupp \sigma$.
    \item The same process can be done for (abs$_c$): we have
      $\typing {\Gamma, \vartyping y T'} \Delta {\lbd x t} C
      \nonlinfunc T$, thus we can apply the induction hypothesis to
      $\typing {\Gamma, \vartyping x C, \vartyping y T'} \Delta t T$
      as $x \neq y$, which gives $\typing {\Gamma, \vartyping x C}
      \Delta {\sigma(t)} T$, and we can conclude.
    \item Suppose the root is (rec), then the same process as for
      (abs$_c$) can be applied.
    \item Suppose the root is (app), meaning we have $\typing
      {\Gamma, \vartyping x \bkvarb''} \Delta {t_1} {\bkvar \linfunc
      \bkvarb'}$ and $\typing {\Gamma, \vartyping x \bkvarb ''}
      {\Delta'} {t_2} \bkvar$; by i.h. we have $\typing \Gamma \Delta
      {\sigma(t_1)} {\bkvar \linfunc \bkvarb'}$ and $\typing \Gamma
      {\Delta'} {\sigma(t_2)} \bkvar$, and thus by (app) we can type
      $\typing \Gamma \Delta {\sigma(t_1)\sigma(t_2)} \bkvarb'$ which concludes.
    \item The same can be done for (app$_c$).
    \item Suppose the root is (sup), meaning we have $\typing
      {\Gamma, \vartyping x T'} \Delta {t_i} T$.
      By induction hypothesis, we get $\typing \Gamma \Delta
      {\sigma(t_i)} T$, and by definition of $\perp$, $t_i \perp t_j$
      implies $\sigma(t_i) \perp \sigma(t_j)$, and thus we can type
      by (sup) $\typing \Gamma \Delta {\sum_{i=1}^n \alpha_i \cdot
      \sigma(t_i)} T$, which concludes by the action of $\sigma$ on
      superpositions.
    \item Suppose the root of the derivation is (shape), first remark
      that $x$ belongs to the $\Gamma$ part, because it has no marker
      around it; thus we have $\typing {\Gamma, \vartyping x T'}
      \Delta t \bkvar$, thus by i.h. $\typing \Gamma \Delta {\sigma(t)}
      \bkvar$, and we can type by (shape) $\typing {\Gamma,
      \shapemarker{\Delta}} {} {\shape{\sigma(t)}} \shape \bkvar$,
      which concludes.
    \item Suppose the root is (contr), meaning we have $\typing
      {\Gamma, z : T', \shapemarker{\vartyping x \bkvar}}{\Delta,
      \vartyping y \bkvar} t T$.
      Again, as $z$ comes with no marker and thus $z \neq x$, which
      means that we can treat the case as before, with a non-linear
      context $\Gamma, \shapemarker{\vartyping{x}{\bkvar}}$, and conclude.
    \item Suppose the root is (equiv), thus $\typing {\Gamma,
      \vartyping x T'} \Delta t T$; by i.h. $\typing \Gamma \Delta
      {\sigma(t)} T$.
      One can check that if $t \equiv t'$, then $\sigma(t) \equiv
      \sigma(t')$, which concludes.
  \end{itemize}

  A similar process can be done to prove the second property. It differs
  only by the fact that $x$ is now present in only one subterm context;
  for the other subterm, while some $x$ may be present by a \textshape{}
  construct, replacing it with a well-typed value of same type, one
  can rewrite the typing tree with the new term (possibly by remodeling the
  (contr) rules), and conclude. Else, we have directly $\sigma(t) = t$,
  thus typing is direct. \qedhere
\end{proof}

\begin{corollary}\label{lem:substitutionlemma}
  Let $\typing {\Gamma, \vartyping{x_1}{T_1}, \dots,
  \vartyping{x_m}{T_m}}{\Delta, \vartyping{x_{m+1}}{T_{m+1}}, \dots,
  \vartyping{x_n}{T_n}} t T$ be a well-typed term.
  Then, for any substitution $\sigma$ such that for all $1 \leq i n$,
  $x_i \in \subsupp \sigma$, $\typing \Gamma {\Delta_i} {\sigma(x_i)}
  {T_i}$, we have $\typing \Gamma {\Delta, \Delta_1, \dots, \Delta_n}
  {\sigma(t)} T$.
\end{corollary}

\begin{proof}
  Direct by induction on $n$: the case $n = 1$ is simply the above
  lemma, and for $n=k+1$, any substitution $\sigma$ can be written as
  $\sigma= \sigma' \circ \tau$ with $\tau = \{x_1 \subarrow
  \sigma(x_1)\}$; $\tau(t)$ is typable by the above lemma and removes
  $\vartyping {x_1} A_1$ from the context, and then
  $\sigma'(\tau(t))$ concludes by induction hypothesis. \qedhere
\end{proof}

\begin{definition}
  Let $\typing \Gamma \Delta t T$ be a well-typed term.
  We say that $t$ types and terminates, written t.a.t.,
  if $t \reduces t_1 \dots \reduces t_k \reduces v$,
  $\typing \Gamma \Delta {t_i} T$,
  $\typing \Gamma \Delta v T$.
\end{definition}

\begin{lemma}\label{lem:ortho-sum}
  Let $(t_i)_{1 \leq i \leq n}$ be well-typed terms
  such that $\forall i \neq j, t_i \perp t_j$.
  Let $s = \sum_{i=1}^n \alpha_i \cdot t_i$,
  and $t = \sum_{i=1}^n \beta_i \cdot t_i$
  such that $s, t$ are t.a.t. and $\sum_{i=1}^n \alpha_i \beta_i^* = 0$.
  Then $s \perp t$.
\end{lemma}

\begin{proof}
  As $t_i \perp t_j$, each $t_i$ reduces to a canonical form.
  Let us write them up to completion as
  $t_i \sreduces \sum_{j=1}^m \gamma_{ij} \cdot v_j$.
  By \autoref{lem:reduces-sum},
  $s \sreduces \sum_{j=1}^m
  (\sum_{i=1}^n \alpha_i \gamma_{ij}) \cdot v_j$, and
  $t \sreduces \sum_{j=1}^m
  (\sum_{i=1}^n \beta_i \gamma_{ij}) \cdot v_j$.
  By definition, this is a quasi canonical form, as some phases may be zero, but
  $v_j$ are two by two distinct; proving the orthogonality condition
  on this will imply
  orthogonality for the exact canonical form.
  First, it is easy to see that they all have the same shape,
  as $t_i$ are orthogonal, thus have the same shape.
  Compiling the inner product gives the following:

  \[
    \sum_{j=1}^m \sum_{j'=1}^m (\sum_{i=1}^n \alpha_i
    \gamma_{ij})(\sum_{i=1}^n \beta_i \gamma_{ij'}) \kron{v_j}{v_j'}
    = \sum_{1 \leq i,i' \leq n} \alpha_i \beta_i^* \sum_{j=1}^m
    \gamma_{ij} \gamma_{ij'}^*
    = \sum_{1 \leq i,i' \leq n} \alpha_i \beta_i^*  \kron{i}{i'} = 0
  \]

  where the second equality is obtained by developing each sum, and the third
  by definition of $t_i \perp t_j$, and we conclude by the condition
  on $\alpha_i, \beta_i$. \qedhere
\end{proof}

\begin{lemma}\label{lem:ortho-sum-all}
  Let $(t_i)_{1 \leq i \leq n}$, $(t_i')_{1 \leq i \leq n}$ such that
  $\forall i j, t_i \perp t_j'$.
  Then for any $\alpha_i, \beta_i$, $\sum_{i=1}^n \alpha_i \cdot t_i
  \perp \sum_{i=1}^n \beta_i \cdot t_i'$.
\end{lemma}

\begin{proof}
  Consider $(\tilde{t_i})_{1 \leq i \leq 2n}$ where $\tilde{t_i} =
  t_i$ for $1 \leq i \leq n$,
  and $\tilde{t_i} = t_i'$ for $n < i \leq 2n$, and apply
  \autoref{lem:ortho-sum}. \qedhere
\end{proof}

\begin{lemma}
  \label{lem:subred-equiv}
  Let $\typing \Gamma \Delta t T$ be a pure term.
  Then we have the following typing tree:

  \[
    \begin{prooftree}
      \infer0[(r)]{\typing {\Gamma_0} \Delta s T}
      \infer1[(contr)]{\typing {\Gamma_1} \Delta {\sigma_1(s)} T}
      \ellipsis{}{}
      \infer1[(contr)]{\typing {\Gamma_n} \Delta {\sigma_n(s)} T}
    \end{prooftree}
  \]

  where $n \geq 0$, each $\Gamma_n = \Gamma$, $\sigma_n(s) = t$,
  $r$ is neither (contr) nor (equiv),
  and each $\Gamma_i$, respectively $\sigma_i$, contains one less
  marked variable, contains one more mapping, as per (contr).
\end{lemma}

\begin{proof}
  Let $s$ be a pure term, and $\typing \Gamma \Delta t T$ be any term.
  We prove by induction on typing of $t$ that if $\qt s t \neq 0$,
  then $s$ satisfies the property of the lemma.
  \begin{itemize}
    \item (ax): $t = x$, $\qt s t \neq 0 \implies s = x$, and we can
      type $s$ directly.
    \item (ax$_c$), (ax$_0$) and (ax$_1$) are obtained similarly.
    \item (qcase): $t = \qcase r {t_0} {t_1}$, $\qt s t \neq 0
      \implies s = \qcase q {t_0} {t_1} \wedge r \equiv q$ by
      \autoref{lem:quantity-pure} as $q$ is pure.
      We can type $q$ via (equiv) by typing $r$ first, and then type
      $s$ via (qcase).
    \item (cons): $t = c(t_1, \dots, t_n)$, non-zero quantity implies
      $s = c(s_1, \dots, s_n)$ with $s_i \equiv t_i$.
      We can type $t_i$ by the typing of $t$, then $s_i$ via (equiv),
      then $s$ via (cons).
    \item (match): same as for (qcase).
    \item (abs), (abs$_c$), (rec): same as (ax).
    \item (app): this implies $t = t_1 t_2$, $s = s_1s_2$, with $s_i
      \equiv t_i$, thus we type $s_i$ through (equiv) and conclude.
    \item (app$_c$): same as (app).
    \item (sup): $t = \sum_{i=1}^n \alpha_i \cdot t_i$, $\qt s t =
      \sum_{i=1}^n \alpha_i \qt s {t_i}$.
      As $\qt s t \neq 0$, then there is $1 \leq i \leq n$ such that
      $\qt s {t_i} \neq 0$.
      We can apply the induction hypothesis on $t_i$, as it has the
      same context and type as $t$.
    \item (shape): $t = \shape r$, $s = \shape q$ and $q \equiv r$,
      thus we type $q$ through (equiv) and conclude.
    \item (equiv): $t \equiv t'$, $\qt s t = \qt s t' \neq 0$ through
      \autoref{lem:quantity-equiv}, and we apply i.h. on $t'$ as it
      has the same context and type.
    \item (contr): $t = \sigma(r)$. We can apply the induction
      hypothesis on $r$, then reapply (contr) on the alternative tree.
  \end{itemize}

  Now, we can conclude as by \autoref{lem:quantity-one}, $\qt s s =
  1$; and any obtained tree satisfies the hypothesis. \qedhere
\end{proof}

\begin{lemma}\label{lem:subred-can-ortho}
  Let $\typing {} {\vartyping x \bkvar} s \bkvarb'$, $\sigma_v = \{x
  \subarrow v\}$,
  $\closedtyping v \bkvar$, $\closedtyping {v'} \bkvar$ such that $v
  \perp v'$ and $s$ t.a.t. .
  We also suppose that any term $t$ that terminates in less
  steps than $s$ is also t.a.t. .
  Then $\sigma_v(s) \perp \sigma_{v'}(s)$.
\end{lemma}

Something that we can note is that given $s \sreduces \alpha \cdot
\zket + \beta \cdot \oket$,
then $\qcasedefault s$ will not always reduce to $\alpha \cdot t_0 +
\beta \cdot t_1$,
as ($\equivcontext$) requires a pure term, which is not guaranteed for $s$.
However, we will have $\qcasedefault s \sreduces \alpha \cdot s_0 +
\beta \cdot s_1$,
where $t_i \sreduces s_i$. In particular, proving results of orthogonality
between such terms will guarantee orthogonality between the $t_i$.
Same goes for \textmatch{}.

\begin{proof}
  By induction on the typing of $s$:
  \begin{itemize}
    \item (ax) gives the desired result directly;

    \item (ax$_c$), (ax$_0$) and (ax$_1$) have an empty linear context;

    \item (qcase): $s = \qcasedefault t$; as $s$ terminates,
      so does $t$, in less steps than $s$, thus $t$ reduces to a well-typed
      value by hypothesis. Therefore,
      either $\typing {} {\vartyping x \bkvar} t \qbit$;
      the only pure closed values of type $\qbit$ are $\zket$ and $\oket$, thus
      $\sigma_v(t) \sreduces \alpha \cdot \zket + \beta \cdot \oket$,
      $\sigma_{v'}(t) \sreduces \gamma \cdot \zket + \delta \cdot \oket$,
      and by orthogonality $\alpha\gamma^* + \beta\delta^*  = 0$.
      By definition, $\sigma_v(s) \sreduces \alpha \cdot t_0 + \beta \cdot t_1$,
      $\sigma_{v'}(s) \sreduces \gamma \cdot t_0 + \delta \cdot t_1$,
      which are well-typed as $s$ is t.a.t.,
      and we can conclude by \autoref{lem:ortho-sum}.
      Else, $\sigma_v(t) = t$, thus $t \sreduces \alpha \cdot \zket +
      \beta \cdot \oket$,
      $\sigma_v(s) \sreduces \alpha \cdot \sigma_v(t_0) + \beta \cdot
      \sigma_v(t_1)$.
      In particular, $t_0 \perp t_1$, thus $\sigma_v(t_i) \perp
      \sigma_{v'}(t_{1-i})$,
      and $\sigma_v(t_i) \perp \sigma_{v'}(t_i)$ by induction
      hypothesis, thus we conclude by \autoref{lem:ortho-sum-all},
      as the obtained summation is t.a.t. as a reduced of $s$.

    \item (cons): there is $1 \leq i \leq n$ with $\typing {}
      {\vartyping x \bkvar} {t_i} \bkvarb_i$.
      As one coordinate is orthogonal and the others are identical,
      orthogonality is verified for the whole term.

    \item (match): same as for (qcase). The only difference is that
      $\sigma(t)$ will reduce to a linear combination of same
      constructor, thus $\sigma_v(s) \sreduces \sum_{j=1}^m \beta_j
      \cdot \sigma_{w_j}(t_i)$,
      where $w_j \perp w_j'$, and the result is direct by induction hypothesis
      and by concluding with \autoref{lem:ortho-sum-all}.

    \item (abs), (abs$_c$), (rec) do not yield a term of type $\bkvar$.

    \item (app): by typing, reduction of $t_1t_2$ implies
      that $t_1$ reduces too to a well-typed closed value,
      \ie either $\lbd x t$ or \textletrec{}.
      Furthermore, by typing,
      either $x$ is in the context of $t_1$, thus of $t$, or in
      the context of $t_2$, thus of $v$.
      In any case, by joining all substitutions, we reach a point
      where $s \sreduces \tilde \sigma(t)$, and we can apply the i.h. on $t$,
      as it is t.a.t. as a reduced from $s$. Same goes for (app$_c$).

    \item (sup): all $t_i$ have the $x$ in the context, thus
      $\sigma_v(t_i) \perp \sigma_{v'}(t_i')$,
      and $\sigma_v(t_i) \perp \sigma_{v'}(t_j)$ for $i \neq j$ as
      $t_i \perp t_j$, thus we conclude by \autoref{lem:ortho-sum-all},
      as we apply the lemma on $s$, which is t.a.t. itself.

    \item (shape): no linear context.

    \item (equiv): apply the hypothesis on $t$, then via (Equiv) we
      can obtain a reduction for $s$.

    \item (contr): just an alpha-renaming of a classical variable,
      does not disturb the previous proofs. \qedhere
  \end{itemize}
\end{proof}

\begin{lemma}[Strong Subject Reduction]
  \label{lem:subred-strong}
  Let $\typing \Gamma \Delta t T$ be a well-typed term,
  such that $t$ terminates in $k$ steps.
  Then $t$ t.a.t. .
\end{lemma}

\begin{proof}
  This is proven by induction on the depth of termination of $t$.
  Let us first prove that $t$ reduces to well-typed terms. If $k =
  0$, the result
  is direct; suppose $k = k' + 1$, thus $t \reduces t' \mostreduces{k'} v$.
  We will prove that $t$ t.a.t. by induction on the derivation
  of $\reduces$. In most cases, we obtain that $t'$ is
  well-typed, which implies directly that $t$ t.a.t. by
  induction hypothesis on $t'$.

  First, note that apart from (Can) and (Equiv), left terms are pure,
  thus we can apply \autoref{lem:subred-equiv} to gain a term with no
  (equiv) rule.
  Now, let us first suppose that there is no (contr) rule near the
  end of the typing tree.

  \begin{itemize}
    \item (Qcase$_0$): the only non (contr) and (equiv) rule typing
      this syntax is (qcase),
      thus $\Delta = \Delta_1, \Delta_2$ with $\typing \Gamma
      {\Delta_1} \zket \qbit$ and $\typing \Gamma {\Delta_2} {t_0} \qvar$.
      Again, as $\zket$ is pure, the only no (contr) and (equiv) rule
      typing $\zket$ is (ax$_0$).
      Therefore, $\Delta_1 = \emptyset$, and we can conclude by typing of $t_0$.

    \item (Qcase$_1$): same as above.

    \item (Match): the root of typing is (match), thus $\Delta =
      \Delta_1, \Delta_2$, with $\typing \Gamma {\Delta_1}
      {c_i^{s_i}(p_v^1, \dots, p_v^{n_i})} \bvar$
      and $\typing {\Gamma, \overrightarrow{y_i} :
      \overrightarrow{C_i}} {\Delta_2,
      \overrightarrow{z_i} : \overrightarrow{Q_i}} {t_i} \bkvarb$.
      Again, by typing as the $v_i$ are pure, this gives us $s_i = \sign
      {\bkvar_1, \dots, \bkvar_n} \bvar$, and $\typing \Gamma
      {\Delta_i^c} {p_v^i} \bkvar_i$,
      with $\Delta_1 = \Delta_1^c, \dots, \Delta_{n_i}^c$.
      As $\overrightarrow{x_i} = \overrightarrow{y_i},
      \overrightarrow{z_i}$, we conclude using
      \autoref{lem:substitutionlemma}.

    \item (Lbd) : If the root is (app), then $\Delta = \Delta_1,
      \Delta_2$, $\typing \Gamma {\Delta_1} {\lbd x s} T' \linfunc T$
      and $\typing \Gamma {\Delta_2} v T'$.
      Typing of $\lbd x s$ implies $\typing \Gamma {\Delta_1,
      \vartyping x T'} s T'$, and thus by substitution lemma on
      $\sigma(s)$ we conclude.
      The same process can be done if the root is (app$_c$).

    \item (Rec) is identical to (Lbd), up to some small tweaks.

    \item (Can) : we consider an alternative typing tree where all
      equivalence relations are at the end, as we make them commute
      with summation, \textqcase{}, \textmatch{} and constructors.
      Therefore, we have $t \equiv s$, where $s$ has no (equiv) rule
      near the end.
      If $s$ is pure, then $t = 1 \cdot s$ (by quantity), by i.h.
      this reduces to
      $1 \cdot s'$ where $s'$ is well-typed, and using (equiv) we can
      type $1 \cdot s'$.
      Else, $s$ is not pure.
      Let us first prove that given two well-typed closed terms $s, t$
      such that $s \perp t$ and that they terminate, reducing to well-typed
      closed value,
      then developing linearly preserves the
      orthogonality, meaning:

      \begin{itemize}

        \item $\qcase s {t_0} {t_1} \perp \qcase t {t_0}
          {t_1}$: we have $s \sreduces v$ and $t \sreduces w$ (they
          terminate as we can derive orthogonality).
          By definition, $v, w$ are well-typed closed values,
          thus $v \equiv \alpha \cdot \zket + \beta \cdot
          \oket$ and $w \equiv \gamma \cdot \zket + \delta \cdot \oket$.
          As they are orthogonal, then they will verify $\alpha
          \gamma^* + \beta \delta^* = 0$ (even if this is not the
            canonical form as some coefficients can be zero, this will
          still hold).
          Now, this indicates, by \evalcontext, that $\qcase s
          {t_0} {t_1} \sreduces \alpha \cdot s_0 + \beta \cdot s_1$
          and $\qcase t {t_0} {t_1} \sreduces \gamma \cdot s_0 +
          \delta \cdot s_1$, where $t_0 \sreduces s_0$ and $t_1 \sreduces s_1$.
          By \autoref{lem:ortho-sum}, as $t_0 \perp t_1$, then
          $s_0 \perp s_1$, and thus
          $\alpha \cdot s_0 + \beta \cdot s_1 \perp\gamma \cdot s_0 +
          \delta \cdot s_1$.

        \item $c(s, \dots, t_n) \perp c(t, \dots, t_n)$: we suppose
          that $s, t$ are in first slot, other cases are done equivalently.
          By \evalcontext, when reducing each term to a value, the
          first slot will contain the reduced value of $s$ and $t$.
          Then, as they are orthogonal, and the other values are
          identical, they have the same shape
          by~\autoref{lem:canonical-shape}, and when writing the
          summation to test the zero equality,
          we can just develop through the first coordinate to obtain
          back the summation from the orthogonality of $s$ and $t$.
        \item $\match s \dots \perp \match t \dots$: by definition,
          $s \sreduces \sum_{i=1}^n \alpha_i \cdot w_i$ and $t
          \sreduces \sum_{i=1}^n \beta_i \cdot w_i$,
          by writing their canonical form, and then completing by
          some $0 \cdot w_i$ in order to have the same pure elements
          on each side.
          As they are orthogonal, then $\sum_{i=1}^n \alpha_i \beta_i^* = 0$.
          Furthermore, they must have the same shape, and be a value
          of type $\bvar$, thus any pure value satisfies $w_i =
          c_k(v_i^1, \dots, v_i^n)$.
          As $w_i \neq w_j$ for $i \neq j$, then one coordinate is distinct.
          Now, $\match s \dots \sreduces \sum_{i=1}^n \alpha_i \cdot
          s_i$ and $\match t \dots \sreduces \sum_{i=1}^n
          \beta_i \cdot t_i$, where $\sigma(w_i) \sreduces s_i$,
          and $\sigma(w_i) \sreduces t_i$.
          For any $i \neq j$, $\sigma_i$ and $\sigma_j$ have map one
          variable to a distinct pure value, as explained above;
          and any $v \neq w$ where $v, w$ are pure values satisfy $v \perp w$.
          We can thus apply \autoref{lem:subred-can-ortho}, and get
          $\sigma_i(w_i) \perp \sigma_j(w_j)$.
          Finally, we can conclude by \autoref{lem:ortho-sum},
          and as $s_i, t_i$ are reduced terms from $\sigma_i(w_i)$,
          orthogonality is preserved.
      \end{itemize}

      Now, as $s$ is not pure, then there is a subterm of $s$ that
      contains a (sup) typing rule, followed by some (maybe $0$) linear rules,
      \ie{} (cons), (match) or (qcase), as they are the only rules
      which need a pure term inside their construction to be pure.
      By definition of $\equivcontext$, we can develop linearly this
      sum to obtain $s' = \sum_{i=1}^n \alpha_i \cdot t_i$;
      and by construction, we can use the above result by induction,
      as each term is a subterm of a terminating term, thus reduces in
      less steps than the original term, therefore we can apply the
      induction hypothesis. Note that some terms may obtain a phase $0$
      at the end, and could have reduced in more steps; but as they will
      be removed, not applying this process is ok, as we will not
      test orthogonality with them anyway.
      as the base case is (sup) thus requires orthogonal terms.
      Therefore, by syntax, $s'$ is a sum of orthogonal well-typed terms.
      We can reapply this process on each $t_i$ until they are pure,
      which terminates as $t_i$ has a smaller typing tree than $s$,
      as one (sup) rule has been removed.
      At the end, we obtain multiple summation of pure orthogonal
      values $\sum_{i=1}^n \alpha_i \cdot (\dots \cdot \sum_{k=1}^l
      \beta_k \cdot t_k)$.
      Each sum can be typed through (sup), as the coefficients come
      from a (sup) rule and the above cases show that orthogonality is kept.
      We can then develop everything, group equivalent pure terms,
      and obtain the canonical form back.
      However, after reducing each term, we can do the reverse and
      rewrite the obtained term as multiple summations.
      Furthermore, each term has either reduced, or stayed the same
      as it was removed;
      by induction hypothesis on $t_i \zoreduces t_i'$, $t_i'$ is
      well-typed with the same context and type as $t$;
      and the orthogonality conditions are still true as we reduce
      terms, thus we can still type (sup) for each sum, and therefore
      type the big sum;
      as this sum is equivalent to $\sum_{i=1}^n \alpha_i \cdot
      t_i'$, by (equiv), we can type it and conclude.

    \item (Shape$_0$) : by typing, we have $\typing {\Gamma'}
      {\Delta'} \zket T'$ where $\Gamma = \Gamma,
      \shapemarker{\Delta'}, \Delta = \emptyset$ and $T = \shape{T'}$.
      Again, as $\zket$ is pure, then we need to type it with
      (ax$_0$), which gives $\Delta' = \emptyset$ and $T = \shape
      \qbit = \unit$.
      We conclude as we can type $\typing \Gamma {} \vunit \unit$
      through (cons).

    \item (Shape$_1$) : same as above.

    \item (Shape$_c$) : as the root is (shape), typing gives $\Gamma
      = \Gamma', \shapemarker{\Delta'}$, $\Delta = \emptyset$, $T =
      \shapemarker{\bvar}$ and $\typing \Gamma \Delta {c(v_1, \dots,
      v_n)} \bvar$.
      Again by typing, the root being (cons), we have $\typing \Gamma
      {\Delta_i} {v_i} \bkvar_i$ and $\Delta = \Delta_1, \dots, \Delta_n$.
      Therefore, (shape) gives $\typing {\Gamma,
      \shapemarker{\Delta_i}} {} {\shape{v_i}} {\shape{\bkvar_i}}$;
      as $\shapemarker{\Delta} = \shapemarker{\Delta_1}, \dots,
      \shapemarker{\Delta_n}$, the definition of $\tilde c$ concludes.
    \item (Shape$_s$): by \autoref{lem:canonical-typed}, each $v_i$
      is well-typed with the same type as the summation, and we
      conclude as above.

    \item (Equiv): By (equiv), $\typing \Gamma \Delta {t_1} T$, by
      i.h., $\typing \Gamma \Delta {t_1'} T$, and by (equiv) again on
      $t_1 \equiv t'$, we conclude.

    \item (\evalcontext): by i.h. on $t \reduces t'$, $\typing \Gamma
      \Delta {t'} T$, thus one can replace the typing tree of $t$ by
      the typing tree of $t'$, and any rule will still hold.
      This is in particular true because $\evalcontext$ is made of no
      superpositions, so none of these checks happen and
      thus typing is preserved.
  \end{itemize}

  Now, given a term with some (contr) rules, we may consider the following:

  \begin{itemize}
    \item One can first remark that any (contr) rule commutes with a
      (qcase), (match), or any (app) rule.
      Therefore, for the (Qcase), (Match), (Lbd), (Rec)
      rules, we can consider an alternative typing tree for $t$ with
      $k \geq 0$ (contr) rules at the end,
      and thus have $t = \sigma_1 \dots \sigma_k(s)$, where $s$ has
      no (contr) rule near the end, thus we can apply the hypothesis above.
      Also note that if $t \reduces t'$, then $t' = \sigma_1 \dots
      \sigma_k(s')$ and $s \reduces s'$. We can then apply the result
      above on $s \reduces s'$,
      and reapply the $k$ (contr) rules on $s'$ to obtain the
      well-typedness of $t'$.
    \item The (Can) rule contains terms that all have the same
      context, thus any (contr) rule may actually be applied on each
      $t_i$ before.
    \item Any (Shape) rule contains only values, thus no marked
      variable is actually given by the context and no (contr) rule is needed.
  \end{itemize}

  Therefore, the result still holds for any term. \qedhere
\end{proof}

\subred*

\begin{proof}
  If $t$ terminates, then we can apply \autoref{lem:subred-strong}.
  If $t$ is pure, one can remark that the proof in \autoref{lem:subred-strong}
  for pure terms does not require termination, as (Can) is stopped from
  going further and needing the termination when $t$ is pure, thus one
  can adapt the proof to get the result.
\end{proof}

\section{Proofs for Section~\ref{subsec:orthogonality}}

\orthoundecide*

\begin{proof}
  Let us first prove the $\pitwo$-hardness.
  Let us recall the definition of Programming Computable
  Functions (PCF), which is defined with the following types and terms grammar:

  \[
    \begin{aligned}
      A, B &\Coloneqq \nat \gmid A \to B \\
      t, t_1, t_2 &\Coloneqq x \gmid \lbd x t \gmid t_1 t_2 \gmid
      \mathtt{fix}~t \gmid \mathtt n \mid \mathtt{succ} \gmid \mathtt{pred}
      \gmid \mathtt{ifz}~t_1 t_2
    \end{aligned}
  \]

  Each term has a corresponding typing rule, and thus PCF is a typed language.
  It is easy to see that \hyrql{} can encode PCF: all of PCF's constructions are
  already present in \hyrql{}, apart from the following terms:
  \[
    \begin{aligned}
      \mathtt{fix}~t &= \letrec f x t\,f\,x \\
      \mathtt{pred} &= \lbd x \match x {0 \to 0, S~n \to n} \\
      \mathtt{ifz} &= \lbd x \lbd y \lbd n \match n {0 \to x, S~m \to y}
    \end{aligned}
  \]
  Therefore, any term of PCF can be encoded into a term in \hyrql{}
  through a total computable function.
  The Universal Halt problem on PCF is known to be $\pitwo$-hard.
  We define $\uhalt$ as the set of terms of PCF, thus terms
  in our language, terminating over any input,
  meaning $t t_1 \dots t_n$ terminates.

  Given a well-typed term $t$ from PCF of type $\tvar_1 \anyfunc
  \dots \anyfunc \tvar_n \anyfunc \bkvar$
  (such type can always be obtained by the Lemma~\ref{lem:functype}),
  we define
  \[
    g(t) = (tx_1\dots x_n, \zket), (tx_1\dots x_n, \oket),
  \]
  where $x_1, \dots, x_n$ are neither free nor bound variables of $t$.
  If $t \in \uhalt$, then it terminates over any input.
  In particular, this implies that $tx_1 \dots x_n$ terminates for
  any substitution, reducing to a value $v$;
  and as $(v, \zket) \perp (v, \oket)$, $g(t) \in \portho$.
  If $g(t) \in \portho$, then $t$ must terminate for any substitution
  of $x_1, \dots, x_n$, thus terminates for any value.
  Therefore, $t \in \uhalt$.
  Thus, $t \in \uhalt \iff g(t) \in \portho$, $\uhalt \leq
  \portho$, and as $\uhalt$ is $\pitwo$-hard, it concludes.

  The fact that $\portho \in \pitwo$ is direct by definition.
  Given $t, t'$ be two terms of the same type with $\fv t = \fv{t'} =
  \{x_1, \dots, x_n\}$, and let $\sigma_{v_1, \dots, v_n} = \{x_i
  \subarrow v_i\}_{1 \leq i \leq n}$, then orthogonality decision can
  be written up as:

  \[
    \begin{aligned}
      \forall v_1, \dots \forall v_n, \exists k \in \mathbb
      N,\ &\sigma_{v_1, \dots, v_n}(\shape t) \mostreduces{k} v_s,
      \sigma_{v_1, \dots, v_n}(\shape{t'}) \reduces^{\leq k} v_s'  \\
      &\sigma(t) \sreduces \sum_{i=1}^n \alpha_i \cdot v_i \in
      \can, \sigma(t') \sreduces \sum_{j=1}^m \beta_j \cdot w_j
      \in \can \\
      & v_s = v_s' \wedge \sum_{i=1}^n \sum_{j=1}^m \alpha_i
      \beta_j^* \kron{v_i}{w_j} = 0
    \end{aligned}
  \]

  The inside property is decidable: checking if a term reducing in
  less than $k$ steps to a variable, written $t \mostreduces{k} v$,
  is decidable; checking the first equality is decidable as there is
  no superposition (computed in $\mcal O(n)$ with $n$ the size of
  $v$), and the equality to $0$ of the double sum and equality is
  decidable as we have restricted the scalars in $\bar{\mathbb Q}$.
  Thus the inside property is $\Pi_0^0$; by definition, $\portho$
  belongs to $\pitwo$, and thus is $\pitwo$-complete. \qedhere

\end{proof}

\ortholower*

\begin{proof}
  One can check that any value of type $T$ is of size $\size v \leq \size T_d$,
  by definition of the depth, therefore we use $n$ as a bound on the size of
  any value of such type. In order to compute and decide $\portho$, we first
  need to reduce $s$ and $t$ to values, which is done in $f, g$ steps. Then,
  computing $\shape s$ and $\shape t$ is done linearly in the size of the
  value, same for computing equality, thus this is done in $\mcal O(n)$ steps.
  For the computation of the equality, one needs to compute the canonical forms
  for the reduced values of $s$ and $t$. In particular, we can complement each
  canonical form, so that testing $\kron{v_i}{w_j}$ is done directly by
  comparing the index. This implies checking all terms of the canonical form,
  and comparing it with all the other terms, thus checking $\mcal O(n)$,
  comparing it with $\mcal O(n)$ terms, and comparison is done in the worst
  case in $\mcal O(n)$ steps. Then, we compute each $\alpha_i \beta_j^*$, and
  then we have to sum $n$ terms and checking nullity. As we have assumed to
  reduce the set of complex phases such that this can be done polynomially, we
  assume it is not the operation with the most cost. Such calculus needs to be
  done in the worst case $n^3$ times.
\end{proof}

\section{Proofs for Section~\ref{subsec:poly}}

\circuitbound*

\begin{proof}
  Let $w \in \shapeset Q$, and $k = \size w$. We produce $\mcC_w^t$
  inductively on the syntax of $t$, first with no approximation, and show that
  $C_w^t$ is of size $P(T(k))$, where $P \in \mbN[X]$.
  From that, we approximate each gate, using the
  chosen universal set of gates, up to a precision $1 - \epsilon$, such that the
  overall circuit has a precision $\frac 2 3$. This approximation, done on each
  gate, can be done using Solovay-Kitaev theorem~\cite{Kit97}, which is known to
  approximate the original gate to a precision $1 - \epsilon$ in time
  $\bigo{\log^c(1/\epsilon)}$, for a constant $c$. Choosing $\epsilon = \frac 2
  {3P(T(k))}$ satisfies the overall precision~\cite{NC12}, and thus, our
  approximated circuit is of size $P(T(k)) \log^c(\frac {3 P(T(k))} 2) =
  \bigo{P(T(k)) \log^c(P(T(k)))}$, which can be bounded by a polynomial $Q \in
  \mbN[X]$, evaluated in $P(T(k))$. Therefore, the circuit is of size
  $Q(P(T(k))) = R(T(k))$ for $R = Q \circ P \in \mbN[X]$.

  Note that due to the fact that $t \in \circuittype$, one can
  always assume that there is no (contr), (equiv) or (shape) rule; by typing,
  any subterm is also always well-typed.
  We may also omit $w$ and write $\mcC^t$.

  All obtained polynomials, by induction, are only summations of polynomial,
  thus there is no exponential blowup on the coefficients. Furthermore, one can
  check that the introduced ancillary qubits that play a role must be applied
  onto a gate, thus we have at most $P(T(k))$ ancillary qubits.

  \begin{itemize}
    \item If $t = \qcasedefault s$, we first put $\mcC_w^s$. Then, either $t_i$
      are values, and this represents an isometry, thus a unitary (up to
      ancillary qubits) which will be approximated later; or they are
      not values, thus of the shape $c(\ket i, s_i)$; and we put the circuits
      $\mcC_w^{s_i}$, controlled (on $\ket i$ for $\mcC_w^{s_i}$), one after
      the other, the control wire being the output of $\mcC_w^s$. As the
      three circuits have a size bounded by a polynomial, the total size is
      bounded by a polynomial too.
    \item If $t = \matchdefaultcollapsed s$, as we know that $t$ terminates, and
      the input has a precise shape, the match will only reduce to a
      specific branch; thus, we put $\mcC_w^s$, followed by $\mcC_w^{t_i}$, for
      the corresponding branch. It also binds the wires labelled by
      the variables in $c(\seq{x_i})$ to the corresponding outputs of $t$. This
      yield one swap, between at most $P(T(k))$ wires, thus can be implemented
      by at most $P(T(k))$ CNOT gates.
    \item Constructors: $0$ generates a classical wire, $S(t)$ adds $1$ to the
      classical wire, $\nil$ is an empty circuit, and $h::t, h \times t$ are
      both circuits put one on top of the other. The bound on the size is
      obtained directly, as for \textqcase{}.
    \item If $t = \lbd x s$, this will only reroute the wire
      labelled with $x$ as the first wire of the circuit (i.e., a swap). if $x$
      is not used, this wire is traced out. The swap can be implemented, as for
      \textmatch{}, by $P(T(k))$ CNOT gates.
    \item Suppose $t = \letrec f x s$. We generate recursively $\mcC_s$, which
      contains boxes for each possible call of $f$. We do any circuit
      simplification (e.g.\ if there is a controlled gate on a wire with value
      $\ket 0$, remove the gate), to remove ill-paths. We have $n \leq \size t$
      boxes corresponding to $f$ in $\mcC^s$. However, one can notice that
      the boxes are controlled on the same wire, by restriction of the syntax
      of the \textqcase{}, and are applied on the same wire, also by restriction
      of \textletrec{}. This problem is called \emph{branch
      sequentialization}, and
      we can transform this into only one gate $f$, adding $\bigo 1$ gates
      before and after such gate \cite[Theorem 9]{HPS25}. Then, we generate
      $\mcC_{w'}^s$, replacing the box $f$ in our circuit, but for the
      new corresponding input shape $w'$. Such process will terminate, as $t$
      terminates, thus there is only a finite number of recursive calls.
      In the worst case, we have added $T(k)$ circuits, thus the total size is
      $T(k)P(T(k)) = P'(T(k))$, for $P' = XP \in \mbN[X]$.
    \item Other terms are quite direct: $x, \zket, \oket$ generate a wire with
      this value, $st$ just puts $\mcC^t$ and then $\mcC_{w'}^s$, where $w'$ is
      the shape of $t$ for the specific input $w$; due to the type restriction,
      we can assume there is no $\textshape{}$ construct, else it would just be
      on natural numbers, which gives the same number back; superpositions are
      only on values, thus will only create ancillary wires. \qedhere
  \end{itemize}
\end{proof}

\circuitexistence*

\begin{proof}
  The circuit can be generated similarly to the proof of
  Theorem~\ref{thm:circuit-bound}, as \boundedrec{} only plays a role for the
  size of the circuit.
\end{proof}

\fbqpsound*

\begin{proof}
  We need to prove that $(C_n^t)_{n \in \mbN}$ is a uniform polynomially-sized
  family of circuits. As $t \in \hyrqlpoly$, there is a polynomial $Q$ such
  that $t$ terminates in time $Q$. Therefore, by
  Theorem~\ref{thm:circuit-bound}, the circuit is of size $\bigo{P(Q(n))}$,
  thus of polynomial size in $n$. Note that we have indexed the circuits by $n$
  rather than $w$, but $\size w = n$, as discussed before. The fact that this
  family is uniform is obtained by building the circuit as done in the proof of
  Theorem~\ref{thm:circuit-bound}: one can build a Turing machine, taking as
  input the size $n$ (thus the shape), and suppose that it contains intially
  the encoding of $t$, thus without the input, onto a tape of the QTM. As the
  circuit is build inductively on the syntax, and it differs only by the input
  shape, this creation is uniform, and is done in polynomial time, as the
  generated circuit is of polynomial size.
\end{proof}

\fbqpcomplete*

\newcommand{\tohyrql}[1]{\computes{#1}}
\newcommand{\pket}{\ket \phi}
\newcommand{\kbk}[3]{|#1\rangle\!\langle#2|#3\rangle}
\begin{proof}
  We prove this result by using Yamakami's algebra~\cite{Yam20}. This paper
  defines a function algebra, and then proves that it completely characterize
  \fbqp{}. Formally, it defines the function class $\square_1^{\mathrm{QP}}$,
  which is the smallest class of functions including the following gates, for
  $\theta \in [0, 2\pi) \cap \tilde{\mathbb C}$:
  \[
    \begin{aligned}
      I(\pket) &\triangleq \pket \\
      PHASE_\theta(\pket) &\triangleq \kbk 0 0 \phi + e^{i\theta}
      \kbk 1 1 \phi \\
      ROT_\theta(\pket) &\triangleq \cos \theta \ket \phi + \sin \theta (\kbk
      1 0 \phi - \kbk 0 1 \phi) \\
      NOT(\pket) &\triangleq \kbk 1 0 \phi + \kbk 0 1 \phi \\
      SWAP(\pket) &\triangleq
      \begin{cases*}
        \pket & if $l(\pket) \leq 1$\\
        \sum_{a, b \in \set{0,1}} \kbk{ab}{ba}{\phi} & otherwise
      \end{cases*}
    \end{aligned}
  \]
  and is closed under the following schemes:
  \[
    \begin{aligned}
      Compo[g, h](\pket) &\triangleq g \circ h(\pket) \\
      Branch[g, h] &\triangleq
      \begin{cases*}
        \pket & if $l(\pket) \leq 1$ \\
        \zket \otimes g(\braket{0|\phi}) + \oket \otimes
        h(\braket{1|\phi}) & otherwise
      \end{cases*} \\
      kQRec_t[g,h,p|\mcal F_k](\pket) &\triangleq
      \begin{cases*}
        g(\pket) & if $l(\pket) \leq t$ \\
        h(\sum_{w \in \set{0,1}^k} \ket w \otimes
        f_w(\braket{w|p(\pket)})) & otherwise
      \end{cases*}
    \end{aligned}
  \]
  Here, $\pket$ is a quantum state, i.e., belongs to the Hilbert space $\mathbb
  C^{2^n}$ for a given $n$ written in Dirac notation. Its size $l(\pket)$ is
  $n$. Given $w \in \set{0,1}^n$ with $n \leq l(\pket)$, we may write
  $\pket = \sum_i \alpha_i \ket{w_i z_i}$, where $w_i \in \set{0,1}^n$ and
  $z_i \in \set{0,1}^{l(\pket) -n}$; we then abuse the notation by writing
  $\braket{w|\phi} = \sum_i \alpha_i \braket{w|w_i} \ket{z_i}$. $\tilde{\mathbb C}$
  denotes the complex numbers whose real and imaginary parts can both be approximated
  by a polynomial-time Turing machine. Note that this falls in the restriction that 
  we did in Section~\ref{subsec:orthogonality}.

  This class is proven to be \fbqp{} complete, thus any function of \fbqp{} can
  be written as a function $\square_1^{\mathrm{QP}}$. We therefore associate
  any function $\square_1^{\mathrm{QP}}$ with a given term in \hyrql{},
  $\tohyrql{\square_1^{\mathrm{QP}}}$, and prove inductively that they belong
  in $\hyrqlpoly$.

  Given a word $w = w_1 \dots w_n\in \set{0,1}^n$, we denote $L(w)$ as the list
  $\ket{w_1}:: \dots :: \ket{w_n} :: \nil$. One can remark that for any
  function $f \in \square_1^{\mathrm{QP}}$, and any input $w \in \set{0,1}^*$,
  $f\ket w = v$, where $\tohyrql{f}L(w) \sreduces v$, thus they compute the
  same function. As $\tohyrql{f}$ is approximated by the family of circuits
  $\mathtt C(\tohyrql f)$ by Theorem~\ref{thm:circuit-bound}, then $f$ is also
  approximated by this family, with precision $\frac 2 3$, thus giving us the
  wanted result.
  \[
    \begin{aligned}
      \tohyrql{I} &\triangleq \lbd x x \\
      \tohyrql{Ph_\theta} &\triangleq \lbd x \qcase x \zket
      {e^{i \theta} \oket} \\
      \tohyrql{Rot_\theta} &\triangleq \lbd x \qcase x {\cos
      \theta \zket + \sin \theta \oket}{-\sin \theta \zket + \cos
      \theta \oket} \\
      \tohyrql{Not} &\triangleq \lbd x \qcase x \oket \zket \\
      \tohyrql{SWAP} &\triangleq \lbd x \match x {
        \nil \to \nil,
        h::t \to \match t {
          \nil \to h::\nil,
          h'::t' \to h'::h::t
        }
      } \\
      \tohyrql{COMP[f, g]} &\triangleq \lbd x
      (\tohyrql{f}(\tohyrql g\,x)) \\
      \tohyrql{Branch[f, g]} &\triangleq \lbd x \match x {
        \nil &\to \nil \\
        h::t &\to \match t {
          \nil &\to h::\nil \\
          h'::t' &\to \qcasesplit h {
            \zket \otimes \tohyrql{f}(h'::t')
          } {
            \oket \otimes \tohyrql{g}(h'::t')
          }
        }
      }
    \end{aligned}
  \]

  All terms belong in $\circuittype \cap \boundedrec$ by construction, and as
  functions have no recursive construct, they terminate in a fixed time,
  independently of the input, and thus they belong in $\hyrqlpoly$.

  For the recursive function, we need to introduce before some constructs.
  $\mathtt{proj}_t$ is of type $\typelist \qbit \linfunc \qbit^{\times t} \times
  \typelist \qbit \times \mathrm{bit}$; it takes a list $x$, then returns the
  first $t$ elements, the remaining part of the list, and a bit indicating if
  the length of the list is strictly greater than $t$. If there is less
  elements, we put $\zket$ by default for the first $t$ elements of the tuple.
  We then create two functions $\mathtt{check}_t$ which returns the list
  reconstructed and the bit, and $\tilde{\mathtt{proj}_t}$ which discards the
  bit. We also define $\textqcase_k (x_1, \dots, x_k) \{\ket s \to t_s\}$ as a
  qcase taking $k$ qubits as inputs and producing all possible $2^k$ branches,
  where $s \in \{0,1\}^k$ is the path taken by each $x_i$; $\seq{\ket s}$ is
  the notation which transform the $k$ qubits as a list of $k$ elements. This
  constructs allow then to define the last function of the algebra, with the
  following set $\mcal F_k = \{f_s : s \in \{0,1\}^k\}$, where each $f_s \in
  \{kQRec_t[g,h,p | \mcal F_k], I\}$, and $\tohyrql{f_s}$ is respectively $f$
  or $\lbd x x$.
  \[
    \begin{aligned}
      \mathtt{proj}_t &\triangleq q \lbd x \match x {
        \nil &\to (\zket, \dots, \zket, \nil, 0) \\
        h_1::r_1 &\to \textmatch\,r_1\dots \match{r_t} {
          \nil &\to (h_1, \dots, h_t, \nil, 0) \\
          h_{t+1}::r_{t+1} &\to (h_1, \dots, h_t, h_{t+1}::r_{t+1}, 1)
        }
      } \\
      \tilde{\mathtt{proj}_t} &\triangleq \lbd x \match {\mathtt{proj}_t} {
        (h_1, \dots, h_d, r, b) \to (h_1, \dots, h_d, r)
      } \\
      \mathtt{check}_t &\triangleq \lbd x \match {\mathtt{proj}_t} {
        (h_1, \dots, h_d, r, b) \to (h_1 :: \dots :: h_t::r, b)
      } \\
    \end{aligned}
  \]
  \[
    \begin{aligned}
      \tohyrql{kQrec_t[g,h,p | \mcal F_k]} &\triangleq \letrec f x
      \textmatch\, {
        \mathtt{check}_t\ x
      } \{
        (y, b) \to \textmatch\, b \{
          0 \to f(y), \\
          1 &\to \match{\tilde{\mathtt{proj}_k}(h\,y)} {
            (h_1, \dots, h_k, r) \to g(\textqcase_k (h_1, \dots, h_k)) \{
              \ket s \to \seq{\ket s} :: \tohyrql{f_s}\,r
            \}
          }
        \}
      \}
    \end{aligned}
  \]

  Again, such encoding belongs to $\circuittype \cap \boundedrec$.
  Furtheremore, terminating in polynomial time is guaranteed, as:
  \begin{itemize}
    \item By induction, each $g,h,p$ computes in polynomial time;
    \item Either $f_s$ is the identity, thus stops here, or is a recursive
      call on an input of size that decreases by $k$;
    \item Therefore, the compute time is roughly, for an input of size
      $n$, $\frac n k \max(P_g(n),P_h(p),P_p(n))$, where $P_g, P_h,P_p$ is the
      compute time of respectively $g,h,p$, which is polynomial, by
      induction hypothesis. The obtain time is thus polynomial in $n$. \qedhere
  \end{itemize}
\end{proof}

\end{document}